\documentclass[letterpaper,onecolumn,10pt,accepted=2021-12-23]{quantumarticle}
\pdfoutput=1
\usepackage[utf8]{inputenc}
\usepackage[english]{babel}
\usepackage[T1]{fontenc}
\usepackage{amsmath}
\usepackage{hyperref}

\usepackage{tikz}
\usepackage{lipsum}

\usetikzlibrary{shapes,arrows,positioning,automata,backgrounds,calc,er,patterns}
\usepackage{tikz-feynman}
\usepackage{physics}
\usepackage{amsmath}
\usepackage{amssymb}
\usepackage{complexity}
\usepackage{float}
\usepackage{hyperref}
\usepackage{bbold}
\usepackage{amsmath}             % for equation typesetting
\usepackage{bbm}
\usepackage{color}               % for colored fonts
\usepackage{setspace}            % for 1.5 and double spacing
\usepackage{graphicx}            % main graphics package
\usepackage{qcircuit}
\usepackage{subfig}
\usepackage{mathtools}
\DeclarePairedDelimiter{\ceil}{\lceil}{\rceil}

\usepackage{amsmath, braket, mathtools, amsthm, amsfonts, xcolor}
\newtheorem{theorem}{Theorem}

\newcommand{\an}{a}
\newcommand{\cre}{\an^\dagger}

\usepackage[numbers,sort&compress]{natbib}

\begin{document}

\title{Simulating Effective QED on Quantum Computers}

\author{Torin F. Stetina}
\affiliation{Department of Chemistry, University of Washington, Seattle, Washington, USA}
\affiliation{Simons Institute for the Theory of Computation, University of California, Berkeley, California, USA}
\email{torins@berkeley.edu}
%\homepage{http://quantum-journal.org}
%\orcid{0000-0002-6279-8455}
%\thanks{You can use the \texttt{\textbackslash{}email}, \texttt{\textbackslash{}homepage}, and \texttt{\textbackslash{}thanks} commands to add additional information for the preceding \texttt{\textbackslash{}author}. If applicable, this can also be used to indicate that a work has previously been published in conference proceedings.}
\author{Anthony Ciavarella}
%\orcid{0000-0003-0290-4698}
\affiliation{Institute for Nuclear Theory, University of Washington, Seattle, Washington, USA}
\affiliation{Department of Physics, University of Washington, Seattle, Washington, USA}
\author{Xiaosong Li}
\affiliation{Department of Chemistry, University of Washington, Seattle, Washington, USA}
%\orcid{0000-0003-1985-4623}
\author{Nathan Wiebe}
\affiliation{Department of Physics, University of Washington, Seattle, Washington, USA}
\affiliation{Pacific Northwest National Laboratory, Richland, Washington, USA}
\affiliation{Department of Computer Science, University of Toronto, Toronto, Ontario, Canada}
\affiliation{Challenge Institute for Quantum Computation, University of Washington, Seattle, Washington, USA}
%\orcid{0000-0002-0335-9508}
\maketitle

\begin{abstract}
In recent years simulations of chemistry and condensed materials has emerged as one of the preeminent applications of quantum computing, offering an exponential speedup for the solution of the electronic structure for certain strongly correlated electronic systems.  To date, most treatments have ignored the question of whether relativistic effects, which are described most generally by quantum electrodynamics (QED), can also be simulated on a quantum computer in polynomial time.  Here we show that effective QED, which is equivalent to QED to second order in perturbation theory, can be simulated in polynomial time under reasonable assumptions while properly treating all four components of the wavefunction of the fermionic field.  In particular, we provide a detailed analysis of such simulations in position and momentum basis using Trotter-Suzuki formulas.  We find that the number of $T$-gates needed to perform such simulations on a $3D$ lattice of $n_s$ sites scales at worst as $O(n_s^3/\epsilon)^{1+o(1)}$ in the thermodynamic limit for position basis simulations and $O(n_s^{4+2/3}/\epsilon)^{1+o(1)}$ in momentum basis.  We also find that qubitization scales slightly better with a worst case scaling of $\widetilde{O}(n_s^{2+2/3}/\epsilon)$ for lattice eQED and complications in the prepare circuit leads to a slightly worse scaling in momentum basis of $\widetilde{O}(n_s^{5+2/3}/\epsilon)$.  We further provide concrete gate counts for simulating a relativistic version of the uniform electron gas that show challenging problems can be simulated using fewer than $10^{13}$ non-Clifford operations and also provide a detailed discussion of how to prepare multi-reference configuration interaction states in effective QED which can provide a reasonable initial guess for the ground state.  Finally, we estimate the planewave cutoffs needed to accurately simulate heavy elements such as gold.
\end{abstract}

\section{Introduction}
Since their inception in the early 1980s~\cite{manin1980computable}, quantum computers have been a highly anticipated technology for simulating the laws of physics. Richard Feynman, known mainly for his work contributing to the development of quantum electrodynamics (QED), the quantum field theory of electromagnetism, was also a pioneer of the idea to use real quantum systems for computation~\cite{feynman1982simulating}. Bringing these two major scientific ideas together, we take a look at how an effective version of QED can be simulated on a universal quantum computer with applications to many body material and molecular systems.

A great body of literature exists that shows that quantum computers can provide exponential advantages for certain \emph{ab initio} electronic structure calculations~\cite{lanyon2010towards,whitfield2011simulation,peruzzo2014variational,reiher2017elucidating,von2020quantum,lee2020even}.  The utility of quantum based devices stems from their ability to directly simulate the dynamics of the Coulomb Hamiltonian without making the approximations that are usually required to make their classical counterparts efficiently tractable.  However, despite these existing approaches to quantum chemistry, often times they still fall short of being theoretically correct, since the Hamiltonian is still approximate.  For example, most approaches to solving the electronic structure problem utilize the Born-Oppenheimer approximation, which assumes that the nuclear wavefunction is uncoupled from the electronic wave function.  While this approximation has received substantial attention in recent quantum simulation work~\cite{kivlichan2017bounding,babbush2019quantum}, comparably less progress has been made to properly treat relativistic effects in electronic structure simulations~\cite{gerritsma2010quantum,fillion2017algorithm},  which give rise to orbital contractions/expansions, spin-orbit coupling, modification of electron-electron repulsion and more ~\cite{reiher2014relativistic, dyall2007introduction}.

Much work has been done on the quantum simulation of relativistic field theories including but not limited to the simulation of scalar field theories \cite{Jordan_2012,Preskill1,Preskill3,Brennen_2015,Marshall_2015,Ciavarella_2020,Barata:2020jtq,Klco:2018zqz,klco2021hierarchical}, fermionic field theories \cite{Preskill2,Lamm_2020,Mazza_2012}, lattice gauge theories \cite{Klco:2018kyo,Schwinger1,Schwinger2,Schwinger3,Schwinger5,Schwinger6,LGT1,LGT2,su2sim,Zohar_2013,Zohar_2015,Lamm_2019,Alexandru_2019,Ba_uls_2020,Tagliacozzo_2013,Tagliacozzo_2013_2,PhysRevLett.105.190404,Tavernelli20_094501,Byrnes_2006,Ciavarella_2021,kan2021lattice,stryker2021shearing,Paulson_2021,davoudi2021simulating,Zohar:2012ay,Zohar:2012xf,Banerjee:2012xg,Banerjee:2012pg,Martinez2016,Muschik:2016tws,Zohar:2016iic,Banuls:2017ena,Kaplan:2018vnj,Zache:2018jbt,Stryker:2018efp,Raychowdhury:2018tfj,Davoudi:2019bhy,Haase:2020kaj,Davoudi:2020yln,Buser:2020cvn,raychowdhury2020solving,Raychowdhury:2019iki,Tagliacozzo:2012df,Ji:2020kjk,Chandrasekharan:1996ih,Brower:2003vy,Wiese:2006kp,atas2021su2,klco2021standard,dejong2021quantum,meurice2021theoretical,zohar2021quantum,armon2021photonmediated,andrade2021engineering}, superstring theory \cite{Gharibyan_2021}, $\sigma$ models \cite{Alexandru_2019_2,Singh:2019uwd,Bhattacharya:2020gpm,Hostetler_2021} and light front QFTs \cite{kreshchuk2020quantum,Kreshchuk:2020kcz} (which are known to have many similarities with quantum chemistry as noted by Kenneth Wilson \cite{WILSON199082}). However, the absence of relativity from the majority of existing quantum chemistry simulation algorithms raises both practical and theoretical questions. First, relativistic effects are significant for the ground states of some molecules.  For example, atomic gold and uranium dimers both have significant relativistic effects that need need to be considered to accurately model their ground state. Without the ability to simulate these effects, quantum computers face a fundamental precision limit that prevents the proper physics from being simulated within arbitrary accuracy, thereby limiting the primary advantage of quantum computers. While previous work has addressed relativistic Hamiltonian simulation using the Dirac Hartree-Fock method~\cite{Pittner12_030304}, this relies on the no-pair approximation where only electronic orbitals are used from the mean field, ultimately leading to the same form of the Coulomb Hamiltonian, but with different coefficients during quantum simulation. While this no-pair relativistic Hamiltonian will obviously capture a portion of relativistic effects being ignored by the non-relativistic Coulomb Hamiltonian, the positronic manifold, and subsequently quantum electrodynamical (QED) effects are ignored. Some common consequences of QED effects include the Lamb shift, change in ionization energies, energy shifts in absorption spectra, and more. A review on relativistic and QED effects in atoms and molecules can be found in~\cite{pyykko12_371}. The leading order source of these QED effects in atomic and molecular material systems are typically known as vacuum polarization, and electron self-energy. These effects arise from the realization that the ground state of the molecule is a joint ground state of both the fermionic as well as bosonic fields and is perhaps most striking with the vacuum polarization term which describes the energy contributions from virtual electron/positron productions interacting with the boson field. Going beyond the no-pair approximation for molecular and material systems in general is sparse in the current literature, due to its high computational cost classically and its theoretical complexity. Due to this complexity of including both fermionic and photonic fields for QED corrections in molecular systems, a convenient approximation can be used to neglect the use of photonic modes in exchange for an effective field theory known as effective QED or eQED \cite{Liu15_631, Lindgren13_014108}. By integrating out photonic degrees of freedom, the electronic and positronic interaction contributions are kept without the added complication of tracking photonic modes present in the vacuum.

The second aspect of why we need to consider including relativity stems from the Quantum Complexity-Theoretic Church-Turing Thesis. This belief, widely held in the quantum information community, is that any physically realistic model of computing can be efficiently simulated by a quantum Turing machine. If physics is viewed as providing an analog computational model, and if this statement is correct, then quantum computers must be able to simulate all physical processes in polynomial time.  Thus, understanding whether quantum computers can solve the electronic structure problem in polynomial time in the presence of relativistic effects is important since it sheds light on the validity of the Quantum Complexity-Theoretic Church-Turing Thesis.

In this paper, we focus on a many-body Hamiltonian perspective and show that quantum computers can efficiently sample from the energy eigenvalues for an approximation to quantum electrodynamics known as eQED.  This formulation of quantum electrodynamics is found from expanding the Feynman path integrals that describe quantum electrodynamics to second order from which a Hamiltonian can be derived for only the electronic and positronic degrees of freedom.  We specifically provide algorithms for simulating these dynamics in both position and momentum basis both with and without an external vector potential.  We then analyze the cost of performing such a simulation using Trotter-Suzuki simulation methods both theoretically and numerically for a relativistic version of the free-electron gas (which we denote rellium).  While the rellium model by itself may be considered to be a toy model from an \emph{ab initio} physics simulation perspective, we examine this model for its novel Hamiltonian terms that arise in the context of quantum simulation algorithms, that will be present in more sophisticated relativistic many-body systems.  In addition, we discuss how to prepare multi-reference configuration singles and doubles (MRCISD) approximations to the ground state for eQED which is necessary because the strong correlations present in systems where eQED is needed will often make elementary approximations such as Hartree-Fock inaccurate. Finally, we present a simple cost estimate for simulating a gold atom using planewaves in eQED, and context for future work involving QED and relativistic effects in quantum simulation.

\section{Review of \lowercase{e}QED}
Quantum electrodynamics (QED) is the quantum field theory that describes the electromagnetic interactions of electrons and positrons. QED is formulated in terms of a four component fermionic spinor field, $\psi_a(x)$ that creates electrons and annihilates positrons and a 4-potential, $A_\mu(x)$, that creates and annihilates photons. It should be noted that the interpretation of $\psi_a(x)$ is not the same as the interpretation of the electron field in non-relativistic field theory. In non-relativistic field theory, applying the electron field operator to a state simply removes a single electron, but in the relativistic case, applying $\psi_a(x)$ to a state will create a superposition of a state with one less electron and a state with one more positron. For this reason, $\psi_a(x)$ can be interpreted as an operator that increases the electric charge of the state by $1$. For relativistic chemical physics, the primary goal is to simulate the motion of electrons and positrons, therefore it would be convenient to have a formulation of QED without the photon degrees of freedom. In other words, the goal of effective QED is to derive a Hamiltonian $H_{eQED}$, such that
\begin{equation}
    \bra{\phi_f} e^{- i H_{QED} t} \ket{\phi_i} = \bra{\tilde{\phi}_f} e^{- i H_{eQED} t} \ket{\tilde{\phi}_i}
\end{equation}
where $\ket{\phi_i}$ and $\ket{\phi_f}$ are states without free photons propagating, $\ket{\tilde{\phi}_i}$ and $\ket{\tilde{\phi}_f}$ are the corresponding states with the photon field integrated out, $H_{QED}$ is the full QED Hamiltonian, and $H_{eQED}$ has the photon field integrated out. The effective Hamiltonian can be derived in the path integral formulation of QED with the Feynman gauge fixing procedure \cite{Schwartz}, giving
\begin{multline}
\label{eq:QEDInt}
    \bra{\phi_f} e^{- i H_{QED} t} \ket{\phi_i} = \int DA_{\mu}(x) D\bar{\psi}(x) D \psi(x) \phi_f^*(A_{\mu}(x),\bar{\psi}(x),\psi(x)) \phi_i(A_{\mu}(x),\bar{\psi}(x),\psi(x)) \\
    \exp(i \int d^4 x \left( i \bar{\psi}(x) \gamma^\mu \partial_{\mu} \psi(x) - m\bar{\psi}(x) \psi(x) - e \bar{\psi}(x) \gamma^{\mu} \psi(x) A_{\mu}(x) - \frac{1}{2} A^{\mu}(x) \square A_{\mu}(x) \right) ) 
\end{multline}
where $e$ is the elementary charge, $\square = \partial_t^2 - \sum_{i=1}^{3} \partial_{x_i}^2 $, and $ DA_{\mu}(x) D\bar{\psi}(x) D \psi(x)$ is the functional measure for the fields being integrated over. The repeated upper and lower indices corresponds to Einstein summation convention with a spacetime metric using the mostly negative convention ($g_{\mu \nu} = \text{diag}(1,-1,-1,-1)$), i.e.
\begin{equation}
    a^\mu b_\mu = a^0 b^0 - \sum_{i=1}^3 a^i b^i \ \ \ .
\end{equation}
The $\gamma$ matrices act on the Fock space that mixes the components of the particles.  From this perspective, they are akin to operations such as fermionic swaps which are widely used in the quantum chemistry literature~\cite{babbush2018low,hagge2020optimal}. The $\gamma$ matrices, when seen as operators acting on the spinor of fermionic operators $[a^\dagger_{\vec{x},0},a^\dagger_{\vec{x},1},a^\dagger_{\vec{x},2},a^\dagger_{\vec{x},3}]$, can be represented in the Dirac representation as
\begin{eqnarray}
    \gamma^0 & = \hat{1} \otimes \hat{\sigma}_3 \nonumber \\
    \gamma^i & = i \hat{\sigma}_2 \otimes \hat{\sigma}_i
\end{eqnarray}
where $\hat{\sigma}_i$ are the standard Pauli matrices with $i=1,2,3$ specifying a spatial direction. In Eq.~\eqref{eq:QEDInt}, $\bar{\psi}(x) = \psi^\dagger(x) \gamma^0$, where the summation over the spinor indices of $\psi^\dagger(x)$ and $\gamma^0$ has been suppressed. For states where the electromagnetic field is close to satisfying the classical equations of motion, an action for eQED can be obtained by performing the integral over $A_\mu(x)$ in the stationary phase approximation yielding
\begin{multline}
    \bra{\tilde{\phi}_f} e^{- i H_{eQED} t} \ket{\tilde{\phi}_i} = \int D\bar{\psi}(x) D \psi(x) \tilde{\phi}_f^*(\bar{\psi}(x),\psi(x))
    \tilde{\phi}_i(\bar{\psi}(x),\psi(x))\\
   \exp\left( i \int d^4 x \left(  i \bar{\psi}(x) \gamma^\mu \partial_{\mu} \psi(x) - m\bar{\psi}(x) \psi(x)
		 - \frac{1}{2} e^2 \bar{\psi}(x) \gamma^{\mu} \psi(x) \square^{-1}(x,y) \int d^4 y \ \bar{\psi}(y) \gamma_{\mu} \psi(y) \right) \right) \ .
\end{multline}
Therefore the action for eQED is given by
\begin{equation}
    S_{eQED} = \int d^4 x \left( i \bar{\psi}(x) \gamma^\mu \partial_{\mu} \psi(x) - m\bar{\psi}(x) \psi(x)
		 - \frac{1}{2} e^2 \bar{\psi}(x) \gamma^{\mu} \psi(x) \square^{-1}(x,y) \int d^4 y \bar{\psi}(y) \gamma_{\mu} \psi(y) \right)
\end{equation}
where $\square^{-1}(x,y)$ is the Green's function for $\square$.
\begin{equation}
\square^{-1}(\vec{x},t',\vec{y},t) = \frac{1}{4\pi \abs{\vec{y} - \vec{x}}} \delta(t' - t - \abs{\vec{y} - \vec{x}})  
\end{equation}
This resulting action is not local in time, which makes it unsuitable for deriving a Hamiltonian for the canonical quantization formulation of quantum mechanics. If the radiation effects are neglected (ie. the Coulomb interaction between electrons and positrons is approximated as being instantaneous), then
\begin{equation}
    \square^{-1}(\vec{x}',t',\vec{x},t) \approx \frac{1}{4\pi \abs{\vec{x}' - \vec{x}}} \delta(t' - t)\label{eq:Gapprox}
\end{equation}
In this approximation,
\begin{equation}
    S_{eQED} = \int d^4 x \left( i \bar{\psi}(x) \gamma^\mu \partial_{\mu} \psi(x) - m\bar{\psi}(x) \psi(x)
	- \frac{1}{2} e^2 \bar{\psi}(x) \gamma^{\mu} \psi(x)  \int d^3 y \frac{1}{4\pi \abs{\vec{x} - \vec{y}}} \bar{\psi}(y) \gamma_{\mu} \psi(y) \right)
	\label{eq:eQEDAction}
\end{equation}
At this point, a Legendre transform can be performed to obtain a Hamiltonian for eQED given by

\begin{equation}
\begin{split}
H_{eQED} = \int d^3 x & \left(   -i \sum_{j=1}^3 \bar{\psi}(x) \gamma^j \triangledown_{j} \psi(x) + m\bar{\psi}(x) \psi(x) \right.\\
 & \,\,\,\, \left. \vphantom{\sum_{j=1}^3} + \frac{1}{2} e^2 \bar{\psi}(x) \gamma^{\mu} \psi(x)  \int d^3 y \frac{1}{4\pi \abs{\vec{x} - \vec{y}}} \bar{\psi}(y) \gamma_{\mu} \psi(y) \right) \ .
\end{split}
\end{equation}
where $\triangledown_{j}$ is the typical 3-space gradient operator. Note that when working in momentum space instead of position space, this issue of time nonlocality is avoided by obtaining the potentials by matching Feynman diagrams at $O(e^2)$ with QED. By performing this matching, the leading order corrections from the radiation effects neglected in the position space formulation will be included in the momentum space formulation.

\subsection{eQED for Free Electrons}
The first case that we will present is the Hamiltonian for eQED for the uniform electron gas.  The non-relativistic version of this model is known as jellium, which not only is useful for the foundations of density functional theory but also has become a standard benchmark problem for quantum chemistry simulation~\cite{babbush2018low}.  Here we will present a formulation that is appropriate for eQED which we call ``rellium'' in analogy to the non-relativistic jellium.  There are two representations that we will consider for the Hamiltonian: position space and momentum space.  Both approaches have different advantages and disadvantages and without further knowledge, it is unclear which will prove to be superior for a given problem without knowing the number of lattice sites in the real-space or reciprocal-space lattice needed for the simulation.

\subsubsection{Lattice eQED}
The Hamiltonian from the previous section can be placed on a discrete cubic lattice with $n_s$ sites and side-length $L$.  In this discrete representation we make the identification that $\psi = \an\sqrt{n_s/L^3}$ and $\bar{\psi} = \cre \gamma^0 \sqrt{n_s/L^3}$ where $\an$ is the standard dimensionless fermionic annihilation operator, and $n_s$ is the total number of lattice sites.  $\psi$ and $\an$ obey the standard fermionic anti-commutation relations
\begin{align}
 \{ \psi_{\vec{x},i}, \psi_{\vec{y},j} \} &= 0 \nonumber\\
 \{ \psi^\dagger_{\vec{x},i}, \psi_{\vec{y},j} \} &= \frac{n_s}{L^3} \delta_{\vec{x},\vec{y}} \, \delta_{i,j}\nonumber \\
\{ \an_{\vec{x},i}, \an_{\vec{y},j} \} &= 0 \nonumber\\
\{ \an^\dagger_{\vec{x},i}, \an_{\vec{y},j} \} &= \delta_{\vec{x},\vec{y}} \, \delta_{i,j}
\end{align}
where the first index labels the position on the lattice and the second index labels the spinor component. The discretized Hamiltonian on the lattice can then be expressed as
$$H_{eQED} = H_{K} + H_m + H_V$$
$$H_m = \frac{L^3}{n_s} \sum\limits_{\vec{x}} m \bar{\psi}_{\vec x} \psi_{\vec x} = \sum_{\vec{x}} m\cre_{\vec x} \gamma^0 \an_{\vec x}$$
\begin{equation}
H_V = \frac{L^6}{n_s^2} \sum\limits_{\vec{x} \neq \vec{y}} \sum_{\mu=0}^3 g_{\mu \mu}\frac{e^2n_s^{1/3}}{8 \pi L \abs{\vec{x} - \vec{y}}} (\bar{\psi}_x \gamma^\mu \psi_x) (\bar{\psi}_y \gamma^\mu \psi_y)    = \sum\limits_{\vec{x} \neq \vec{y}} \sum_{\mu=0}^3 g_{\mu \mu}\frac{e^2n_s^{1/3}}{8 \pi L\abs{\vec{x} - \vec{y}}} (\cre_x \gamma^0 \gamma^\mu \an_x) (\cre_y \gamma^0 \gamma^\mu \an_y) 
\end{equation}
Here we use the convention that $\vec{x}$, $\vec{y}$ and $\vec{p}$ are vectors of integers that index a particular fermionic mode. The state of the system can be represented by using a qubit to represent each $a^\dagger_x a_x$. This will require $4n_s$ qubits total.

A technicality emerges in choosing the correct quantization of the kinetic operator on the lattice, $H_K$.  In particular, if a finite difference on the lattice is used as in
\begin{equation}
H_{\text{na\"ive}} = \frac{L^{2}}{n_s^{2/3}} \sum\limits_{\vec{x}} -i \bar{\psi}_x \sum\limits_{j=1}^{3} \gamma^j \frac{\psi_{x + \hat{j}} - \psi_{x - \hat{j}}}{2}  
\end{equation}
where $\hat{j}$ is a unit vector pointing in the $j$-th direction, a continuum limit will not be recovered \cite{NNTheorem1,NNTheorem2}. Consider a free electron with integer valued momentum $p$, in units
of $2\pi/L$ in one-dimension. The electron's momentum then lies in the range $[-\frac{n_s^{1/3}}{2},\frac{n_s^{1/3}}{2}]$where the energy is given by $E = \sqrt{m^2 + \frac{n_s^{2/3}}{L^2} \sin^2({2\pi p/n_s^{1/3}})}$. Therefore with this kinetic Hamiltonian, both an electron with momentum $\frac{n_s^{1/3}}{2}$ and an electron with zero momentum have energy $m$. Likewise, by periodicity, at every energy there will be twice as many states as there are in the continuum limit. This doubling of states prevents this na\"ive choice of Hamiltonian from correctly reproducing the physics of continuum eQED. Several solutions to this fermion doubling problem have been developed in the study of lattice gauge theories \cite{Doubling1,Doubling2,Doubling3,Doubling4,KAPLAN1992342,Neuberger_1998_1,Neuberger_1998_2}. One solution, known as SLAC fermions, solves this problem by choosing the kinetic term, $H_K$ such that the dispersion relation agrees with the continuum limit \cite{Doubling3,Doubling4}.

\begin{equation}
H_K \rightarrow H_{SLAC} :=  \frac{2\pi L^2}{n_s^{1/3}}\sum\limits_{\vec{x},\vec{y},\vec{p}}  \frac{e^{i {2 \pi n_s^{-\frac{1}{3}}}\vec{p} \cdot (\vec{x} - \vec{y})}}{n_s} \bar{\psi}_{\vec{y}}  \gamma^j  p_j \psi_{\vec{x} }= \frac{2\pi }{n_s^{1/3}L}\sum\limits_{\vec{x},\vec{y},\vec{p}}  {e^{i {2 \pi n_s^{-\frac{1}{3}}}\vec{p} \cdot (\vec{x} - \vec{y})}}\cre_{\vec{y}} \gamma^0 \gamma^j  p_j \an_{\vec{x} }
\end{equation}

This choice of kinetic term solves the fermion doubling at the expense of locality. The cost model we use for the simulation assumes all-to-all couplings between the qubits and can perform CNOT between such pairs as well as all single qubit Clifford operations. Therefore, the loss of locality does not represent a significant cost in this model. Other solutions exist, such as  Wilson's kinetic operator \cite{Doubling1}, domain wall fermions \cite{KAPLAN1992342}, and overlap fermions \cite{Neuberger_1998_1,Neuberger_1998_2},  but for simplicity we focus on the SLAC case.

\subsubsection{Momentum Space Finite Volume eQED}
The momentum space formulation of a relativistic eQED Hamiltonian without an external potential can be viewed as a variant of a well known interacting electron model: jellium. We will follow similar notation to Ref. \citenum{babbush2018low}, where the standard non-relativistic 3 dimensional jellium Hamiltonian is defined as the following in second quantization
\begin{equation}
H_{jel} = \sum_{p,\sigma_1} \frac{k_p^2}{2} a^{\dagger}_{p,\sigma_1} a_{p,\sigma_1} + \frac{1}{2 L^3} \sum_{\substack{(p,\sigma_1) \neq (q,\sigma_2)  \\ \nu \neq 0}} \frac{4 \pi}{k_\nu^2} a^{\dagger}_{p,\sigma_1} a^{\dagger}_{q,\sigma_2} a_{q+\nu,\sigma_2} a_{p-\nu,\sigma_1}
\end{equation}
where the $p,q$ indices are momentum space electronic planewave orbitals, $L$ is the length of one dimension of the cubic simulation box, $\sigma_1, \sigma_2 \in [\uparrow, \downarrow]$ are the fermion spin indices, and $k$ is the momentum. Here $\{a_{p,\sigma_1}^\dagger, a_{q,\sigma_2}\} = \delta_{p,q} \delta_{\sigma_1,\sigma_2}$ and atomic units are chosen such that the charge on the electron $e$, and the electron mass $m_e$, obeys $e=m_e=1$ to match the standard unit choice in the literature. In non-relativistic jellium, the planewave basis is only defined for electronic and spin degrees of freedom
\begin{align}
\varphi_\nu (r) = \sqrt{\frac{1}{L^3}} e^{i k_{\nu} \cdot r},\qquad
k_\nu = \frac{2 \pi \nu}{L}
\end{align}
where $\varphi_\nu (r)$ is a single plane wave as a function of the momentum grid point $\nu \in [-N^{(1/3)}, N^{(1/3)}]^3 \subset \mathbb{Z}^3$  and distance is denoted by $r$. The momentum grid is the same size for both spin-up and spin-down planewaves.

In order to extend the 3D jellium Hamiltonian to a relativistic framework, we use the 4-spinor solution to the Dirac equation to build a planewave basis, and use the second order tree Feynman diagrams for eQED to compute the two particle interactions. As mentioned above, we will refer to this relativistic jellium model as `rellium'. In general, the planewave basis has the same form as above, but the total number of planewaves will be doubled, now that positronic degrees of freedom must be considered. The general form of the rellium Hamiltonian takes into account relativistic particle energies, and also eQED interaction amplitudes that have pair creation $(e^- \rightarrow e^- e^-e^+)$ interactions in addition to the typical two body interactions $(e^- e^- \rightarrow e^-e^-)$. In this form charge is conserved, but not particle number. The rellium Hamiltonian can be described in second quantization as
\begin{align}
H_{rel} =& \sum_{p,\sigma_1} E_p a^{\dagger}_{p,\sigma_1} a_{p,\sigma_1} + \sum_{p,\sigma_1} E_p b^{\dagger}_{p,\sigma_1} b_{p,\sigma_1}  \nonumber\\
&+ \frac{1}{2 L^3} \sum_{\substack{p,q,r \\ \sigma_1, \sigma_2, \sigma_3, \sigma_4}} \frac{\mathcal{M}_{e^-_{p,\sigma_1}e^-_{p,\sigma_2}}^{e^-_{r,\sigma_3}e^-_{p+q-r,\sigma_4}}}{\sqrt{E_p E_q  E_r E_{p+q-r}  }} a^{\dagger}_{p+q-r,\sigma_4} a^{\dagger}_{r,\sigma_3} a_{q,\sigma_2} a_{p,\sigma_4} \nonumber\\
&+ \frac{1}{2 L^3} \sum_{\substack{p,q,r \\ \sigma_1, \sigma_2, \sigma_3, \sigma_4}} \frac{\mathcal{M}_{e^+_{p,\sigma_1}e^+_{q,\sigma_2}}^{e^+_{r,\sigma_3} e^+_{p+q-r,\sigma_4}}}{\sqrt{E_p E_q  E_r E_{p+q-r}  }} b^{\dagger}_{p+q-r,\sigma_4} b^{\dagger}_{r,\sigma_3} b_{q,\sigma_2} b_{p,\sigma_1} \nonumber\\
&+ \frac{1}{2 L^3} \sum_{\substack{p,q,r \\ \sigma_1, \sigma_2, \sigma_3, \sigma_4}} \frac{\mathcal{M}_{e^-_{p,\sigma_1}e^+_{q,\sigma_2}}^{e^-_{r,\sigma_3} e^+_{p+q-r,\sigma_4}}}{\sqrt{E_p E_q  E_r E_{p+q-r}  }} b^{\dagger}_{p+q-r,\sigma_4} a^{\dagger}_{r,\sigma_3} b_{q,\sigma_2} a_{p,\sigma_1} \nonumber\\
&+ \frac{1}{2 L^3} \sum_{\substack{p,q,p_1 \\ \sigma_1, \sigma_2, \sigma_3, \sigma_4}} \frac{\mathcal{M}_{e^-_{p,\sigma_1}}^{e^+_{q,\sigma_2} e^-_{p_1,\sigma_3} e^-_{p-q-p_1,\sigma_4}}}{\sqrt{E_p E_q  E_{p_1} E_{p-q-p_1}  }} a^{\dagger}_{p-q-p_1,\sigma_4} a^{\dagger}_{p_1,\sigma_3} b^{\dagger}_{q,\sigma_2} a_{p,\sigma_1} + h.c. \nonumber\\
&+ \frac{1}{2 L^3} \sum_{\substack{p,q,p_1 \\ \sigma_1, \sigma_2, \sigma_3, \sigma_4}} \frac{\mathcal{M}_{e^+_{p,\sigma_1}}^{e^-_{q,\sigma_2} e^+_{p_1,\sigma_3} e^+_{p-q-p_1,\sigma_4}}}{\sqrt{E_p E_q  E_{p_1} E_{p-q-p_1}  }}  b^{\dagger}_{p-q-p_1,\sigma_4} b^{\dagger}_{p_1,\sigma_3} a^{\dagger}_{q,\sigma_2} b_{p,\sigma_1} + h.c.  \nonumber\\
&+ \frac{1}{2 L^3} \sum_{\substack{p,q,r \\ \sigma_1, \sigma_2, \sigma_3, \sigma_4}} \frac{\mathcal{M}_{0}^{e^-_{p,\sigma_1} e^+_{q,\sigma_2} e^-_{r,\sigma_3} e^+_{-p-q-r,\sigma_4}}}{\sqrt{E_p E_q  E_r E_{-p-q-r}  }} a^{\dagger}_{p,\sigma_1} b^{\dagger}_{q,\sigma_2} a^{\dagger}_{r,\sigma_3} b^{\dagger}_{-p-q-r,\sigma_4} +h.c.  \nonumber\\
&+ \delta m \sum_{p,\sigma_1,\sigma_2} \frac{1}{2 E_p n_s} \Biggr(\bar{u}_{\sigma_1}(p) u_{\sigma_2}(p) a_{p,\sigma_1}^\dagger a_{p,\sigma_2} - \bar{v}_{\sigma_1}(p) v_{\sigma_2}(p) b_{p,\sigma_1}^\dagger b_{p,\sigma_2} \nonumber\\
&+ \bar{v}_{\sigma_1}(-p) u_{\sigma_2}(p) b^{\dagger}_{-p,\sigma_1} a_{p,\sigma_2} + \bar{u}_{\sigma_1}(-p) v_{\sigma_2}(p) a^\dagger_{-p,\sigma_1} b_{p,\sigma_2}
 \Biggr) +\Lambda n_s \label{eq:Hrel}
\end{align}
where the relativistic energy $E_k$ is defined as $E_k = \sqrt{k_x^2 + k_y^2 + k_z^2 + m^2}$ where $m$ is the electron mass, $a, b$ are the annihilation operators corresponding to electronic and positronic degrees of freedom respectively, $\mathcal{M}$ represents the computed eQED amplitudes for the interaction, and $\sigma_1, \sigma_2, \sigma_3, \sigma_4 \in [\uparrow, \downarrow]$ are all independent spin indices. $\delta m$ is the difference between the bare mass and the physical electron mass and $\Lambda$ is a constant added to guarantee that the vacuum has zero energy. It is necessary to include these terms in the Hamiltonian to guarantee that the particles in this discretized theory have a mass equal to the electron mass and that the correct physics is reproduced in the continuum limit. The electron and positron operators obey the standard fermionic anti-commutation relations
\begin{align}
\{ a_{\vec{p},i}, a_{\vec{q},j} \} &= 0 \nonumber \\
\{ a^\dagger_{\vec{p},i}, a_{\vec{q},j} \} &= \delta_{\vec{p},\vec{q}} \, \delta_{i,j} \nonumber \\
\{ a_{\vec{p},i}, b_{\vec{q},j} \} &= 0 \nonumber \\
\{ b_{\vec{p},i}, b_{\vec{q},j} \} &= 0 \nonumber \\
\{ b^\dagger_{\vec{p},i}, b_{\vec{q},j} \} &= \delta_{\vec{p},\vec{q}} \, \delta_{i,j} \ .
\end{align}
The helicity spinors are defined by
\begin{align}
u_1(p) &=
\sqrt{E_p + m} 
\begin{pmatrix}
1 \\
0 \\
\frac{p_z}{E_p + m} \\
\frac{p_x + ip_y}{E_p + m}
\end{pmatrix}, \
u_2(p) =
\sqrt{E_p + m} 
\begin{pmatrix}
0 \\
1 \\
\frac{p_x - ip_y}{E_p + m} \\
\frac{-p_z}{E_p + m} 
\end{pmatrix}
,\notag \\
v_1(p) &=
\sqrt{E_p + m} 
\begin{pmatrix}
\frac{p_x - ip_y}{E_p + m} \\
\frac{-p_z}{E_p + m} \\
0 \\
1
\end{pmatrix} 
, \
v_2(p) =
\sqrt{E_p + m} 
\begin{pmatrix}
\frac{p_z}{E_p + m} \\
\frac{p_x + ip_y}{E_p + m} \\
 1 \\
 0
\end{pmatrix} 
\end{align} 
and are used in the construction of $\mathcal{M}$. The details of the amplitudes $\mathcal{M}$ are discussed in detail in Appendix \ref{appendix:momentumpotential}. The state of the system can be represented on a quantum computer by using a qubit to represent the value of each $a_{p,\sigma}^\dagger a_{p,\sigma}$ and each $b_{p,\sigma}^\dagger b_{p,\sigma}$.  Thus the total number of qubits required to encode the state of the fermionic field here is $4n_s$, in exact agreement with the number required in position space (despite the fact that in momentum basis we explicitly divide the field into a fermionic and anti-fermionic subsystem). 

\subsection{eQED with an External Potential}
In the continuum, when eQED is done in the presence of an external vector potential $A^{ex}(x)$, an additional term 
\begin{equation}
    H_{ext} = - e \int d^3 x \bar{\psi}(x) \gamma^\mu \psi(x) A^{ex}_{\mu}(x)\label{eq:extpot}
\end{equation}
must be added to the Hamiltonian.  This vector potential term is more general than the Coulomb term commonly used for external potentials in chemistry applications.  In part, this is because it applies to general external electric potentials but also because this term includes the vector potential needed to describe interactions with external magnetic fields as well.

Since quantum computer simulations necessarily require discretized wave functions, it is necessary to consider discretizations of Eq.~\eqref{eq:extpot}.  The two natural approaches to discretize the system in a lattice are in position and momentum representations.  In the position representation, the integral of the external potential can be discretized as a finite sum via
\begin{equation}
    H_{L,ext} = - e \frac{L^3}{n_s} \sum_{x} \bar{\psi}_x \gamma^\mu \psi_x A^{ex}_{\mu,x}.
\end{equation}

The momentum space representation can be found from the position space representation by applying the Fourier transform to the field operators.  This approach is analogous to the fermionic Fourier transform used in quantum chemistry simulations~\cite{verstraete2009quantum,babbush2018low}; however, here the transform needs to be performed over all four components of the field.  The transform of the vector potential operator is given by
\begin{equation}
    A^{ex}_{\mu}(x) = \int \frac{d^3x}{(2 \pi)^3} e^{-i p x} \tilde{A}^{ex}_{\mu}(p),
\end{equation}
This term takes the form
\begin{equation}
\begin{aligned}
    H_{ext} =   - e \sum_{\sigma_1,\sigma_2} \int \frac{d^3 p d^3 q}{2\sqrt{E_p E_q}} & \Biggr(\bar{u}_{\sigma_2}(q) \gamma^{\mu} u_{\sigma_1}(p) a^{\dagger}_{q,\sigma_2} a_{p,\sigma_1} \tilde{A}^{ex}_{\mu}(p-q) \\
    & + \bar{u}_{\sigma_2}(q) \gamma^{\mu} v_{\sigma_1}(p) a^{\dagger}_{q,\sigma_2} b^{\dagger}_{p,\sigma_1} \tilde{A}^{ex}_{\mu}(-p-q) \\
    & + \bar{v}_{\sigma_2}(q) \gamma^{\mu} u_{\sigma_1}(p) b_{q,\sigma_2} a_{p,\sigma_1} \tilde{A}^{ex}_{\mu}(p+q) \\
    & + \bar{v}_{\sigma_2}(q) \gamma^{\mu} v_{\sigma_1}(p) b_{q,\sigma_2} b^{\dagger}_{p,\sigma_1} \tilde{A}^{ex}_{\mu}(q-p) \Biggr) \\
\end{aligned}
\end{equation}
in momentum space. 

In the discretized momentum space simulation, the external potential term takes the form of
\begin{equation}
\begin{aligned}
    H_{p,ext} = {}  - e \sum_{\sigma_1,\sigma_2}  \sum_{p,q} \frac{1}{2\sqrt{E_p E_q} L^3} & \Biggr(\bar{u}_{\sigma_2}(q) \gamma^{\mu} u_{\sigma_1}(p) a^{\dagger}_{q,\sigma_2} a_{p,\sigma_1} \tilde{A}^{ex}_{\mu}(p-q) \\
    & + \bar{u}_{\sigma_2}(q) \gamma^{\mu} v_{\sigma_1}(p) a^{\dagger}_{q,\sigma_2} b^{\dagger}_{p,\sigma_1} \tilde{A}^{ex}_{\mu}(-p-q) \\
    & + \bar{v}_{\sigma_2}(q) \gamma^{\mu} u_{\sigma_1}(p) b_{q,\sigma_2} a_{p,\sigma_1} \tilde{A}^{ex}_{\mu}(p+q) \\
    & + \bar{v}_{\sigma_2}(q) \gamma^{\mu} v_{\sigma_1}(p) b_{q,\sigma_2} b^{\dagger}_{p,\sigma_1} \tilde{A}^{ex}_{\mu}(q-p) \Biggr). \\\label{eq:Hpext}
\end{aligned}
\end{equation}
With these additional definitions, we now have a general form of the eQED Hamiltonian that we can use for cost analysis.
%%% here %%%

\section{Trotter-Suzuki Simulations of \lowercase{e}QED}
There are many techniques that have been proposed thus far for simulating quantum dynamics.  The first approach proposed for quantum simulation involves the use of Trotter-Suzuki formulas to compile quantum dynamics into a discrete sequence of gate operations~\cite{lloyd1996universal,zalka1998simulating,berry2007efficient,wiebe2011simulating,su2020nearly}. These approaches are space optimal and can take advantage of properties such as locality and commutation relations that qubitization cannot.  For simulations of the free electron gas, known as jellium, recent work has shown that the scaling of the time complexity of Trotter-Suzuki simulation methods and qubitization are nearly equal.  For this reason, we focus on Trotter-Suzuki simulations.

Trotter-Suzuki simulations can be viewed as a method for compiling the unitary matrix $e^{-iHt}$ as a sequence of unitary gates, $U$, such that $\|e^{-iHt} - U\| \le \delta$, where the notation $\| \cdot \|$ refers to the spectral norm, and $\delta$ is a chosen error threshold.  If $H=\sum_j H_j$ for a set of local Hamiltonians $H_j$, such that $e^{-iH_j \theta}$ can be efficiently compiled as a quantum circuit via
\begin{equation}
    U_2(t):=\left(\prod_{j=1}^m e^{-iH_j t/2} \right)\left(\prod_{j=m}^1 e^{-iH_j t/2} \right)= e^{-iHt} + O\left(\max_{j,k,\ell} \|[H_j,[H_k,H_\ell]]\|t^3\right)
\end{equation}
then this approximation can effectively compile $e^{-iHt}$ into a sequence of unitary operations. Here, $O(\cdot)$ refers to the standard big-O notation denoting an upper bound in the asymptotic limit. Additionally, $\widetilde{O}(\cdot)$, $\Theta(\cdot)$ denotes the asymptotic upper bound with suppressed poly-logarithmic terms, and the asymptotic tight bound respectively used throughout this manuscript.
Since each $H_j$ is assumed to be implementable using a polynomial-sized circuit, this approximation effectively compiles $e^{-iHt}$ into a sequence of unitary operations.
Higher-order variants of the Trotter-Suzuki approximation can be constructed from the symmetric Trotter formula $U_2$ via~\cite{suzuki1990fractal,childs2021theory}
\begin{equation}
\begin{split}
    U_{2k+2}(t) &:= U_{2k}^2(s_{2k}t)U_{2k}((1-4s_{2k})t)U_{2k}^2(s_{2k}t) \\ 
    &\, = e^{-iH t} +O\left(\max_{j_1,\ldots,j_{2k+3}}(\|[H_{j_1},[\cdots[H_{2k+2},H_{2k+3}]\cdots]]\|)t^{2k+3}\right),
\end{split}
\end{equation}
where $s_{2k} = (4-4^{1/(2k+1)})^{-1}$.
Such high-order Trotter-Suzuki formulas are not always superior to their lower-order brethren.  This is because the number of exponentials in $U_{2k}(t)$ is in $\Theta(5^{k}m)$ and hence tradeoffs between the exponential improvements to accuracy yielded increasing $k$ and the exponentially increasing costs of doing so must be made.  Further, as the error in the Trotter-Suzuki approximation depends on the commutators between the Hamiltonian terms, the cost of such simulations can be better than these upper bounds suggest~\cite{childs2018toward,hastings2015improving,childs2021theory}.

Thus in order to construct the operation $U_{2k}(t)$, for some integer value of $k$, we need to develop circuits for implementing each of the terms in the decomposition separately.  That is to say we need to take each Hamiltonian term present in the Hamiltonian and convert them to easily simulatable Hamiltonians before using the Trotter-Suzuki approximation to compile it to a sequence of operations that can be run on a quantum computer.

The individual $H_j$ in our representation will, similar to chemistry, be expressable as Pauli operators through the use of a Jordan-Wigner transformation.  Such a transformation yields the following transformation for the fermionic creation operator $a^\dagger_x$ via
\begin{equation}
    a^\dagger_x \mapsto \frac{(X-iY)_x(\bigotimes_{n<x} Z_n)}{2},
\end{equation}
for some arbitrary canonical ordering of the site labels.
Note that it may be tempting to make this assignment to the field operators $\psi^\dagger$ and $\psi$ but we cannot do so directly since the field operators are dimensionful.  For this reason, we discuss in the following the dimensionless fermionic operators $a_x$ and  will use these operators interchangeably with their anti-particle counterparts, $b_y$.  Other fermionic representations are possible, such as the Bravyi-Kitaev encoding~\cite{seeley2012bravyi}, but here we use Jordan-Wigner for its simplicity.  

One technicality that needs to be considered with the Jordan-Wigner encoding is that the pattern of Pauli-$Z$ operations depends on the lexigraphical ordering of the site labels.  In one-dimension, such orderings are straightforward, but in higher dimensions there are a multitude of natural lexigraphical orderings of the sites (orbitals) that can be chosen.

Here we focus on a simplified cost model for the simulation wherein non-Clifford operations constitute the majority of the cost.  Specifically, we assume our quantum computer has all-to-all couplings between the qubits and can perform CNOT between such pairs as well as all single qubit Clifford operations (which can be formed from products of the Hadamard gate $H$ and the phase gate $S=\sqrt{Z}$).  We also assume that the $T=\sqrt{S}$ gate can be applied to each quantum bit.  Further we assume that all Clifford operations can be implemented without cost and only $T$-gates are costly.  This further motivates why we choose Jordan-Wigner representation for our problem because the additional gates needed to enforce the correct signs from the lexigraphical ordering are all Clifford operations, which we take to be without cost.  Thus, within our cost model, the choice of the ordering of the labels of the sites  will prove to be irrelevant.

With the Jordan-Wigner transformation in place we have all we need to compile the circuit from the Trotter-Suzuki simulation $U_{2k}(t)$ into a sequence of gates that can be executed on a quantum computer.  Below, we discuss ways that exponentials of the one- and two-body terms in the Hamiltonian can be compiled into a gateset involving Clifford gates and single qubit rotations.  We will discuss later what translations need to be done to convert the single qubit rotations into circuits involving $H$ and $T$.

%-----
\subsection{Quantum Circuit for the one-body operators}
The free piece of the eQED Hamiltonian consists of a sum over terms of the form ${\psi}^\dagger_p \psi_q$, however the representation of these terms can vary depending on whether we are interested in the position or momentum basis. In the momentum basis formulation, all terms take the form $a_{\sigma,p}^{\dagger} a_{\sigma,p}$ and $b_{\sigma,p}^{\dagger} b_{s,p}$, which acts only on a single qubit register and is trivial to simulate. Therefore, evolving according to the free Hamiltonian in the momentum basis requires $4n_s$ single qubit gates. In the position space lattice formulation, the free Hamiltonian terms can take one of two more additional forms: $a_x^\dagger a_y +a_y^\dagger a_x$ and $i(a_x^\dagger a_y - a^\dagger_y a_x)$.
These operators can then be converted into Pauli operators using the Jordan-Wigner transformation.  
Assuming without loss of generality that in the canonical ordering we choose $x < y$, the Jordan-Wigner representation of these terms takes the form
\begin{enumerate}
    \item $a_{x}^\dagger a_y + a_{y}^\dagger a_x \xrightarrow{JW} X_y (\bigotimes_{x<n<y} Z_n) X_x + Y_y (\bigotimes_{x<n<y} Z_n) Y_x$
    \item $i(a_{x}^\dagger a_y - a_{y}^\dagger a_x) \xrightarrow{JW} X_y (\bigotimes_{x<n<y} Z_n) Y_x - Y_y (\bigotimes_{x<n<y} Z_n) X_x$
\end{enumerate}
While Case $1$ appears in standard constructions for quantum circuits for simulating chemistry~\cite{whitfield2011simulation}, Case $2$ does not typically arise in existing quantum circuit constructions and so we provide optimized networks for implementing it, shown in Figure~\ref{fig:xy_trotter}.  An optimized circuit for the one-body operations in Case $1$ is given in Figure~\ref{fig:xx_trotter}

\begin{figure}[t!]
    \centering
    \[
    \Qcircuit @C=1em @R=.7em {
& \qw & \targ     & \qw      & \qw      & \qw                       & \ctrl{3} & \qw                       & \ctrl{3} & \qw      & \qw      & \targ     & \qw \\
& \qw & \qw       & \ctrl{1} & \qw      & \qw                       & \qw      & \qw                       & \qw      &\qw      & \ctrl{1} & \qw       & \qw \\ 
& \qw & \qw       & \targ    & \ctrl{1} & \qw                       & \qw      & \qw                       & \qw      &\ctrl{1} & \targ    & \qw       & \qw \\
& \qw & \ctrl{-3} & \gate{H} & \targ    & \gate{e^{i\frac{c}{2}Z}}  & \targ    & \gate{e^{-i\frac{c}{2}Z}} & \targ    &\targ    & \gate{H} & \ctrl{-3} & \qw \\
}
\]
    \caption{Circuit used to implement $e^{i c \left( X_y \left(\bigotimes_{x<n<y} Z_n \right)\otimes X_x + Y_y \left(\bigotimes_{x<n<y} Z_n \right)\otimes Y_x \right)}$} for $y=x+2$.
    \label{fig:xx_trotter}
\end{figure}
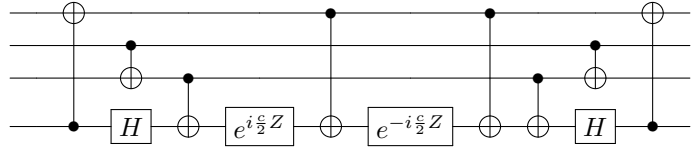

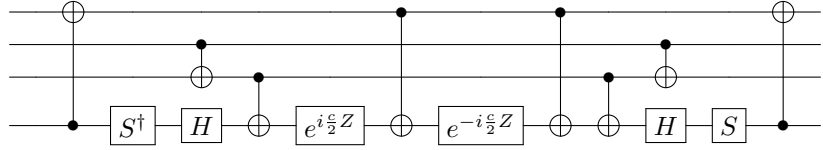
\begin{figure}[t!]
    \centering
    \[
    \Qcircuit @C=1em @R=.7em {
& \qw & \targ     & \qw &\qw      & \qw      & \qw                       & \ctrl{3} & \qw                       & \ctrl{3} & \qw      & \qw      & \qw       & \targ & \qw \\
& \qw & \qw       & \qw & \ctrl{1} & \qw      & \qw                       & \qw      & \qw                       & \qw&\qw      & \ctrl{1} & \qw        & \qw & \qw \\ 
& \qw & \qw       & \qw & \targ    & \ctrl{1} & \qw                       & \qw      & \qw                       & \qw &\ctrl{1} & \targ    & \qw        & \qw & \qw \\
& \qw & \ctrl{-3} & \gate{S^\dagger} & \gate{H} & \targ    & \gate{e^{i\frac{c}{2}Z}}  & \targ    & \gate{e^{-i\frac{c}{2}Z}} & \targ & \targ    & \gate{H} & \gate{S} & \ctrl{-3} & \qw \\
}
\]
    \caption{Circuit used to implement $e^{i c \left(X_y \left(\bigotimes_{x<n<y} Z_n \right) \otimes Y_x - Y_y \left(\bigotimes_{x<n<y} Z_n \right)\otimes X_x \right)}$} for $y=x+2$.
    \label{fig:xy_trotter}
\end{figure}

There are $2n_s (4n_s -1)$ couplings of this form in the Hamiltonian. Therefore, Trotter simulation of the free Hamiltonian requires
\begin{equation}
    N_{rot}= 20n_s(4n_s-1) \label{eq:nrotFree}
\end{equation}
single qubit $Z$-rotations. The focus of this paper is on fault-tolerant quantum simulation and so the cost of simulation is dominated by the number of non-Clifford operations needed to implement the rotations. Other cost models may not assume Clifford operations can be implemented without cost and the number of CNOT gates required may be of interest. Ignoring the Jordan-Wigner strings, $8n_s(4n_s-1)$ CNOTs are needed per step. The Jordan-Wigner strings require 

\begin{equation}
4\sum\limits_{n=0}^{4n_s-1} \sum\limits_{k=0}^{n-1} k = 2\sum\limits_{n=0}^{4n_s-1} n(n-1) = \frac{16}{3}n_s(2n_s-1)(4n_s-1)    
\end{equation}
CNOT gates. Therefore a single Trotter step requires $\frac{8}{3} n_s (16n_s^2-1)$ CNOT gates for the free Hamiltonian. However, the number of CNOT operations can be reduced to $\widetilde{O}(n_s^2)$ using fermionic swap networks~\cite{babbush2018low,o2019generalized,hagge2020optimal}.

\subsection{Interaction Circuits}
\label{section:intcircuits}
As shown in the previous section, the eQED Hamiltonian potential is a sum of a $2\rightarrow2$ term that takes the same form as the non-relativistic case, and a new $1 \rightarrow 3$ term. New circuits will be required to simulate this $1 \rightarrow 3$ term which takes the form $\sum \limits_{j>k>l,m} h_{V \, j,k,l,m} a_j^\dagger a_k^\dagger a_l^\dagger a_m + h.c.$. Using the Jordan-Wigner encoding, $a_j^\dagger a_k^\dagger a_l^\dagger a_m$ can take 4 different forms depending on value of $m$. For each case below, it is assumed that for all other sites greater than the highest site index and smaller than the lowest index, the local $j$th operator is simply the identity matrix $I_j$.
\\
Case 1. $m < l$
\begin{equation}
\begin{split}
&a_j^\dagger a_k^\dagger a_l^\dagger a_m = \\
&-\frac{(X- i Y)_j}{2} \otimes \left(\bigotimes_{n = k+1}^{j-1} Z_n \right) \otimes \frac{(X- i Y)_k}{2} \otimes \left(\bigotimes_{n = l+1}^{k-1} I_n \right) \otimes \frac{(X- i Y)_l}{2} \otimes \left(\bigotimes_{n = m+1}^{l-1} Z_n\right) \otimes \frac{(X + i Y)_m}{2} 
\end{split}
\end{equation}
\\
Case 2. $m \in [l+1,k-1] $
\begin{equation}
\begin{split}
&a_j^\dagger a_k^\dagger a_m a_l^\dagger  =   \\
&\frac{(X - i Y)_j}{2} \otimes \left(\bigotimes_{n = k+1}^{j-1} Z_n \right) \otimes \frac{(X - i Y)_k}{2} \otimes \left(\bigotimes_{n = m+1}^{k-1} I_n \right) \otimes \frac{(X + i Y)_m}{2} \otimes \left(\bigotimes_{n = l+1}^{m-1} Z_n \right) \otimes \frac{(X - i Y)_l}{2} 
\end{split}
\end{equation}
\\
Case 3. $m \in [k+1,j-1]$
\begin{equation}
\begin{split}
&a_j^\dagger a_m a_k^\dagger a_l^\dagger  = \\
&-\frac{(X - i Y)_j}{2} \otimes \left(\bigotimes_{n = m+1}^{j-1} Z_n \right) \otimes \frac{(X + i Y)_m}{2} \otimes \left(\bigotimes_{n = k+1}^{m-1} I_n \right) \otimes \frac{(X - i Y)_k}{2} \otimes \left(\bigotimes_{n = l+1}^{k-1} Z_n \right) \otimes \frac{(X - i Y)_l}{2}
\end{split}
\end{equation}
\\
Case 4. $m > j$
\begin{equation}
\begin{split}
&a_m a_j^\dagger a_k^\dagger a_l^\dagger = \\
&\frac{(X + i Y)_m}{2} \otimes \left(\bigotimes_{n = j+1}^{m-1} Z_n \right) \otimes \frac{-(X - i Y)_j}{2} \left(\bigotimes_{n = k+1}^{j-1} I_n \right) \otimes \frac{-(X - i Y)_k}{2} \otimes \left(\bigotimes_{n = l+1}^{k-1} Z_n \right) \otimes \frac{-(X - i Y)_l}{2} 
\end{split}
\end{equation}
\\

Suppressing the chains of $Z$'s and identities, the contribution to the Hamiltonian takes the following forms
\\
Case 1.
$$H = -\frac{h_V + h_V^*}{16}(XXXX + XXYY + XYXY - XYYX + YXXY - YXYX - YYXX - YYYY) $$
\begin{equation}
- i \frac{h_V - h_V^*}{16}(XXXY - XXYX - XYXX - YXXX - XYYY - YXYY - YYXY + YYYX)
\end{equation}
\\ 
Case 2.
$$H = \frac{h_V + h_V^*}{16} (XXXX + XXYY - XYXY - YXXY + YXYX - YYXX + XYYX -YYYY)$$
\begin{equation}
 + i \frac{h_V - h_V^*}{16}(-XYYY - YXYY + YYXY -YYYX - YXXX - XYXX + XXYX - XXXY)
\end{equation}
\\ 
Case 3.
$$H = - \frac{h_V + h_V^*}{16}(XXXX - XXYY + XYXY - YXXY - YXYX + YYXX + XYYX - YYYY)$$
\begin{equation}
-i \frac{h_V - h_V^*}{16} (-XXXY - XXYX + XYXX - YXXX - XYYY + YXYY - YYXY - YYYX)
\end{equation}
\\Case 4.
$$H =  \frac{h_V + h_V^*}{16}(XXXX - XXYY -XYXY + YXXY - XYYX + YXYX + YYXX - YYYY)$$
\begin{equation}
+i \frac{h_V - h_V^*}{16}(XYYY - YXYY - YYXY - YYYX - XXXY - XXYX - XYXX + YXXX).
\end{equation}

Additionally, the eQED Hamiltonian also contains 4 creation/annhilation terms when expressed in the planewave basis.  The Jordan Wigner representation of these terms is
\begin{equation}
\begin{split}
&a_j^\dagger a_k^\dagger a_l^\dagger a_m^\dagger = \\
&\frac{(X- i Y)_j}{2} \otimes \left(\bigotimes_{n = k+1}^{j-1} Z_n \right) \otimes \frac{(X- i Y)_k}{2} \otimes \left(\bigotimes_{n = l+1}^{k-1} I_n \right) \otimes \frac{(X- i Y)_l}{2} \otimes \left(\bigotimes_{n = m+1}^{l-1} Z_n \right) \otimes \frac{(X - i Y)_m}{2} 
\end{split}
\end{equation}

$$H =  \frac{h_V + h_V^*}{16}(XXXX - XXYY - XYXY - YXXY - XYYX - YXYX - YYXX + YYYY)$$
\begin{equation}
+i \frac{h_V - h_V^*}{16}(XYYY + YXYY + YYXY + YYYX - XXXY - XXYX - XYXX - YXXX).
\end{equation}

From these equations, it can be seen that each term contains every possible tensor product of 4 $X$'s and $Y$'s. In a na\"ive Trotterization simulation, each term would be treated separately. However, this is suboptimal because all terms with an even number of $X$'s and $Y$'s commute and all terms with an odd number of $X$'s and $Y$'s commute. Previous work has introduced a circuit to simulate all terms with an even number of $X$'s and $Y$'s. It will be shown here that the same techniques can be adapted for the odd case as well. A Hamiltonian that contains all terms with odd numbers of $X$'s and $Y$'s can be efficiently simulated by using a circuit that will simultaneously diagonalize all terms.

\begin{figure}[H]
    \centering
    \subfloat{
    \Qcircuit @C=1em @R=.7em {
& \qw & \ctrl{3} & \gate{S} & \gate{H}& \qw  \\
& \qw & \targ & \qw & \qw & \qw \\ 
& \qw & \targ & \qw & \qw & \qw \\
& \qw & \targ & \qw & \qw & \qw \\
}
}
    \caption{Circuit $G$ used to diagonalize all tensor products of an odd number of $X$'s and $Y$'s.}
    \label{fig:odd_diagonal}
\end{figure}
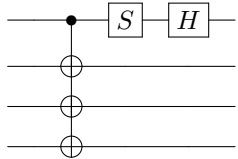

The circuit $G$ can be used to implement time evolution according to a Hamiltonian made out of a sum over the odd tensor products. For example, take
\begin{equation}
    H = XXXY - XXYX - XYXX - YXXX - XYYY - YXYY - YYXY + YYYX.
\end{equation}
Then using the results in Appendix~\ref{appendix:diagonalization} we see that   
\begin{equation}
    H = G (ZZZZ - ZZZ1 - ZZ11 - Z111 + ZZ1Z + Z11Z + Z1ZZ - Z1Z1) G^{\dagger}
\end{equation}
and the resulting circuit in Fig. \ref{fig:odd_pauli} can be used to implement $e^{i \alpha H}$. From Fig. \ref{fig:odd_pauli}, it can be seen that single term in the interaction Hamiltonian requires 12 single qubit gates to implement, and 16 CNOT gates in addition to however many CNOT gates are needed to implement the Jordan-Wigner strings. In the position space formulation, there are at most $256 n_s^2$ terms and in the momentum space formulation there are at most $8192 n_s^3$ terms.

\begin{figure}[t!]
     \centering
\[
    \Qcircuit @C=0.24em @R=.6em {
&\gate{H}&\gate{S^\dagger}&\ctrl{3}&\gate{e^{-i\alpha Z}}&\ctrl{3}&\qw&\qw&\qw&\qw&\qw&\qw&\qw&\qw&\ctrl{3}&\targ&\gate{e^{-i\alpha Z}}&\targ&\gate{e^{-i\alpha Z}}&\targ&\gate{e^{-i\alpha Z}}&\targ&\ctrl{3}&\gate{S}&\gate{H}\\
&\qw&\qw&\targ&\qw&\qw&\qw&\qw&\qw&\ctrl{2}&\qw&\qw&\qw&\ctrl{2}&\qw&\ctrl{-1}&\qw&\qw&\qw&\ctrl{-1}&\qw&\qw&\targ&\qw&\qw&\qw  \\ 
&\qw&\qw&\targ&\qw&\qw&\qw&\ctrl{1}&\qw&\qw&\qw&\ctrl{1}&\qw&\qw&\qw&\qw&\qw&\ctrl{-2}&\qw&\qw&\qw&\ctrl{-2}&\targ&\qw&\qw&\qw \\
&\qw&\qw&\targ&\qw&\targ&\gate{e^{i\alpha Z}}&\targ&\gate{e^{i\alpha Z}}&\targ&\gate{e^{i\alpha Z}}&\targ&\gate{e^{i\alpha Z}}&\targ&\targ&\qw&\qw&\qw&\qw&\qw&\qw&\qw&\targ&\qw&\qw&\qw \\
}
\]

    \caption{Implementation of $e^{i \alpha (XXXY - XXYX - XYXX - YXXX - XYYY - YXYY - YYXY + YYYX)}$}
    \label{fig:odd_pauli}

\end{figure}
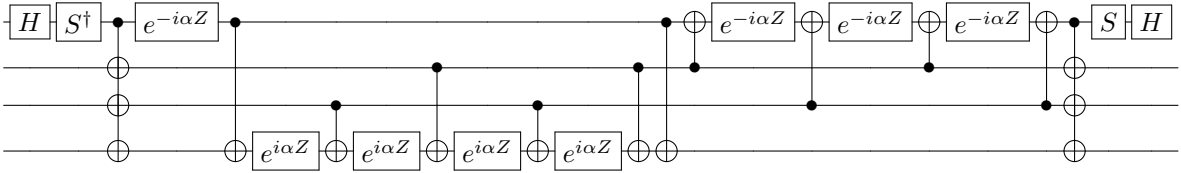

\section{Cost Estimates for \lowercase{e}QED Simulation}
The aim of this section is to provide preliminary cost estimates for simulating effective quantum electrodynamics on quantum computers using Trotter-Suzuki approximations and also provide a comparison to the asymptotic scaling expected from a na\"ive application of qubitization.  All such cost estimates are performed within a computational model wherein Clifford gates are free but non-Clifford gates (specifically the $T$-gate) are not free.  We will consider the cost of simulations both for lattice eQED (position space) or momentum space eQED.
\subsection{Cost Estimates for Lattice eQED}
Estimating the trotterization error for eQED requires the computation of nested commutators of terms in the Hamiltonian~\cite{babbush2015chemical,childs2021theory}. The Hamiltonian for eQED in position space is given by
\begin{equation}\label{eq:Hpos}
    H = H_{SLAC} + H_{m} + H_{int} +H_{L,ext}
\end{equation}
where
$$H_m = \frac{L^3}{n_s}\sum\limits_{\vec{x}} m \bar{\psi}_x \psi_x= \sum_{\vec{x}} m\cre_{\vec x} \gamma^0 \an_{\vec x}$$

$$H_{L,ext} = - e \frac{L^3}{n_s} \sum_{x} \bar{\psi}(x) \gamma^\mu \psi(x) A^{ex}_{\mu}(x)=-\sum_{\vec{x}}e a_{\vec{x}}^\dagger \gamma^0 \gamma^\mu A_\mu^{ex}(x)\an_{\vec{x}} $$

$$
H_{SLAC} = \frac{2\pi }{n_s^{1/3}L}\sum\limits_{\vec{x},\vec{y},\vec{p}}  {e^{i {2 \pi n_s^{-\frac{1}{3}}}\vec{p} \cdot (\vec{x} - \vec{y})}}\cre_{\vec{y}} \gamma^0 \gamma^j  p_j \an_{\vec{x} }=: \sum_{\vec{x},\vec{y}}\sum_{\mu,\mu'}T_{\vec{x},\vec{y}}^{(\mu,\mu')}\cre_{\vec{x},\mu}\an_{\vec{y},\mu'}
$$

\begin{equation}
H_{int} = \sum\limits_{\vec{x} \neq \vec{y}} \sum_{\mu=0}^3 g_{\mu \mu}\frac{n_s^{1/3}e^2}{8 \pi L \abs{\vec{x} - \vec{y}}} (\cre_{\vec{x}} \gamma^0 \gamma^\mu \an_{\vec{x}}) (\cre_{\vec{y}} \gamma^0 \gamma^\mu \an_{\vec{y}})=: \sum_{\vec{x}\ne \vec{y}}\sum_{\mu,\mu'} h^{(\mu,\mu',\nu,\nu')}_{V\, \vec{x},\vec{y}}  \, \cre_{\vec{x},\mu} \an_{\vec{x},\mu'}   \cre_{\vec{y},\nu} \an_{\vec{y},\nu'} ,
\end{equation}
here we have taken the convention that $\vec{x}$ and $\vec{y}$ are $4$-vectors of integers and that the indices $\mu,\nu$ specify one of the components of the $4$-vector.
\begin{theorem}
Let $H$ be the Hamiltonian of Eq.~\eqref{eq:Hpos} with $n_s\ge 8$ sites in the cubic lattice with side-length $L$ and external vector potential operator $A^{ex}(x)$ and let $\Lambda:= \max\left(m+e\max_{x}\|A^{ex}(x)\|,\frac{e^2n_s}{L}, \frac{n_s^{4/3}}{L}\right)$.  Finally let $\ket{\psi}$ be an eigenstate of $H$ such that $H\ket{\psi}= E\ket{\psi}$.   The number of $T$-gates, $N_T$, needed to estimate $E$ within error $\epsilon$ and constant failure probability less than $1/3$ obeys
$$
N_T \in \left(\frac{\Lambda n_s^{2}}{\epsilon} \right)^{1+o(1)}.
$$
\end{theorem}
\begin{proof}
In order to estimate the error in the Trotter-Suzuki approximation, we need to evaluate the commutators between all these terms.  A first step towards this is to estimate the commutators between the individual terms involved.

In order to evaluate the commutator terms involving $H_{SLAC}$ we need to estimate the magnitude of these terms.  As seen above, this involves estimating an oscillating sum.  This expression is symmetric with respect to exchange of $x,y,z$ axis labels and so we will proceed by bounding the coefficient for $k=1$ (i.e. the $x$ component of the momentum.  Further, let $\vec{\Delta} = \vec{x} - \vec{y}$, and $p_k$ be the $k$th component of vector $\vec{p}$. 

\begin{align}
    T_{\vec{x},\vec{y}}^{(1,1)} &=  \frac{2\pi}{n_s^{\frac{1}{3}} L}\sum_{\vec{p}} p_1\left({e^{i {2 \pi n_s^{-\frac{1}{3}}}\vec{p} \cdot \vec{\Delta}}} \right)\nonumber\\
    &=\frac{2\pi}{n_s^{\frac{1}{3}} L}\left( \sum_{p_1=-n_s^{1/3}/2+1}^{n_s^{1/3}/2}p_1 e^{i 2\pi n_s^{-1/3} p_1 \Delta_1}\left(\sum_{p_2,p_3 = -n_s^{1/3}/2+1}^{n_s^{1/3}/2}e^{i 2\pi n_s^{-1/3} (p_2 \Delta_2 +p_3 \Delta_3)} \right) \right)
\end{align}
Next let us assume $\Delta_2>0$ and $\Delta_3>0$.  In this case we have from the formula for the sum of a geometric series that under these circumstances 
\begin{equation}
    \left| \sum_{p_2= -n_s^{1/3}/2+1}^{n_s^{1/3}/2}e^{i 2\pi n_s^{-1/3} (p_2 \Delta_2)}\right|=1+O(n_s^{-1/3}).
\end{equation}
Similarly, if $\Delta_2 =0$ then 
\begin{equation}
    \left| \sum_{p_2= -n_s^{1/3}/2+1}^{n_s^{1/3}/2}e^{i 2\pi n_s^{-1/3} (p_2 \Delta_2)}\right|=n_s^{1/3}.
\end{equation}
We then have that, for $n_s^{1/3} \ge 2$
\begin{align}
    |T_{\vec{x},\vec{y}}^{(1,1)}| &\le\frac{2\pi}{n_s^{\frac{1}{3}} L}\left| \sum_{p_1=-n_s^{1/3}/2+1}^{n_s^{1/3}/2}p_1 e^{i 2\pi n_s^{-1/3} p_1 \Delta_1}\left(\sum_{p_2,p_3 = -n_s^{1/3}/2+1}^{n_s^{1/3}/2}e^{i 2\pi n_s^{-1/3} (p_2 \Delta_2 +p_3 \Delta_3)} \right) \right|\nonumber\\
    &\le \frac{2\pi}{n_s^{\frac{1}{3}} L}\left( \sum_{p_1=-n_s^{1/3}/2+1}^{n_s^{1/3}/2} |p_1| \left|\sum_{p_2,p_3 = -n_s^{1/3}/2+1}^{n_s^{1/3}/2}e^{i 2\pi n_s^{-1/3} (p_2 \Delta_2 +p_3 \Delta_3)} \right| \right)\nonumber\\
    &\in O\left( \frac{1}{n_s^{\frac{1}{3}} L}\left( \sum_{p_1=-n_s^{1/3}/2}^{n_s^{1/3}/2} |p_1| (1+ \delta_{\Delta_2,0}n_s^{1/3})(1+ \delta_{\Delta_3,0}n_s^{1/3}) \right)\right)\nonumber\\
    &\subseteq O\left(\frac{1}{n_s^{\frac{1}{3}} L}\left( \frac{n_s^{1/3}}{2}\left(\frac{n_s^{1/3}}{2} +1 \right) (1+ \delta_{\Delta_2,0}n_s^{1/3})(1+ \delta_{\Delta_3,0}n_s^{1/3}) \right)\right)\nonumber\\
    &\subseteq O\left(\frac{n_s^{1/3}}{L} (1+ \delta_{\Delta_2,0}n_s^{1/3})(1+ \delta_{\Delta_3,0}n_s^{1/3}) \right)
\end{align}
By symmetry, the exact same bound holds by permuting the labels of the indices,
\begin{equation}
    |T_{\vec{x},\vec{y}}^{(\chi,\chi)}| \in O\left(\frac{n_s^{1/3}}{L} (1+ \delta_{\Delta_{\chi+1 ~{\rm mod }~ 3 +1},0}n_s^{1/3})(1+ \delta_{\Delta_{\chi +2~{\rm mod }~ 3 +1},0}n_s^{1/3} )\right) 
\end{equation}
The exact same argument can be applied to find a similar expression for $|T_{\vec{x},\vec{y}}^{(\chi,\xi)}|$ for $\xi \ne \chi$.  Specifically, it can be shown that for each $\chi ,\xi$ there exist $f,g\in  \{1,2,3\}$ such that
\begin{equation}
    |T_{\vec{x},\vec{y}}^{(\chi,\xi)}| \in O\left(\frac{n_s^{1/3}}{L} (1+ \delta_{\Delta_{f},0}n_s^{1/3})(1+ \delta_{\Delta_{g},0}n_s^{1/3} )\right) 
\end{equation}

Next it is straightforward to see that for each of the $n_s$ terms in $H_m$, their coefficients are at most $m$.  The situation for $H_{\rm int}$ is a little more complicated since the coefficients for a given term vary with the distance between $\vec{x}$ and $\vec{y}$.  It is useful for us to envision a probability distribution over the possible coefficients that emerge in the expansion.  Let $V$ be a random variable drawn from a uniform distribution over upper bounds on the coefficients of $T$.  It follows from the above discussion that there exist constants $\kappa_1$ and $\kappa_2$ such that the distribution on the upper bounds on the coefficients for the creation operators in the potential term obeys

\begin{equation}
    P \left(V\le(\kappa_1e^2/L)\right) \in O(1),\qquad P \left(V \ge ( \kappa_2 e^{2}n_s^{1/3}/L) \right) \in O(1/n_s).
\end{equation}
Thus we have that the expectation value of $V$ satisfies
\begin{equation}
    \mathbb{E}(V) \in O \left(\frac{e^2}{L} \right)\label{eq:Vavg}
\end{equation}

Similarly, from the previous discussion it follows that if we define $W$ to be a random variable found by sampling $T^{(\chi,\xi)}_{\vec{x},\vec{y}}$ we find that there exist constants $K_1, K_2,K_3,K_4$ such that the sampled upper bound on the coefficients reads
\begin{equation}
\begin{split}
    P\left(W \le\left(\frac{K_1n_s^{\frac{1}{3}}}{ L } \right)\right) &\in O(1), \\  
    P\left(W \in\left[ \frac{K_2n_s^{2/3}}{L}, \frac{K_3n_s^{2/3}}{L} \right]\right) &\in O(1/n_s^{1/3}), \\
    P\left(W \ge \left(\frac{K_4n_s}{ L} \right)\right) &\in O(1/n_s^{2/3})
\end{split}
\end{equation}
Hence the expectation value of $W$ satisfies
\begin{equation}
    \mathbb{E}(W) \in O\left(\frac{n_s^{1/3}}{L}\right).
\end{equation}
Thus since $W>0$ we have from Chebyshev's inequality that for any constant $\delta>0$, $P(W \ge \delta\mathbb{E}(W)) \in O(1/\delta)$. Thus from the union bound for independent and identically distributed variables $W_1,\ldots,W_N$, $P(W_1\cdots W_N \ge \delta \mathbb{E}(W)^N) \in O(N/\delta)$.  Therefore taking $\delta \in \Theta(n_s^{\lambda N})$ for some constant $\lambda>0$ yields $P(W_1\cdots W_N \ge  (n_s^\lambda \mathbb{E}(W))^N) \in O(N n_s^{-\lambda N})$.  Thus because $W \in O(n_s/L)$ for all inputs,
\begin{equation}
    W^N \in N \mathbb{E}(W)^{N+O(1)}.\label{eq:unionarg}
\end{equation}
The exact same reasoning implies that $V^N$ is similarly bounded.

Now let us consider all commutators consisting of $N_m$ mass terms, $N_{ext}$ external potential terms, $N_T$ kinetic terms and $N_V$ two-body interaction terms.  Each such commutator is a Lie-product between even monomials of creation and annihilation operators of degree at most $4$.  
Each commutator at most doubles the number of terms in the polynomial and increases the degree of each monomial by at most $4$.  For example, we have that if $\phi_{s_1} \phi_{s_2}\phi_{s_3}\phi_{s_4}$ is a monomial taken from the set $\phi\in \{a_1^\dagger, a_1,\ldots, a_{n_s}^\dagger, a_{n_s},1\}$ that for any polynomial in the creation and annihilation operators of degree $d$, $P$, which can be expressed as $P=\sum_j c_j \phi_{\sigma_{j,1}}\cdots \phi_{\sigma_{j,d}}$ we have that
\begin{equation}
\begin{split}
    [\phi_{s_1} \phi_{s_2}\phi_{s_3}\phi_{s_4}, P] &= \sum_j c_j [\phi_{s_1} \phi_{s_2}\phi_{s_3}\phi_{s_4},\phi_{\sigma_{j,1}}\cdots \phi_{\sigma_{j,d}}] \\ 
    &=\sum_{j:\{\phi_s\} \bigcap \{\phi_{\sigma_{j,1}}\cdots \phi_{\sigma_{j,d}}\} \not\subseteq \{\emptyset,1\}} c_j [\phi_{s_1} \phi_{s_2}\phi_{s_3}\phi_{s_4},\phi_{\sigma_{j,1}}\cdots \phi_{\sigma_{j,d}}]
\end{split}
\end{equation}
Therefore there exists a polynomial $Q$ of degree at most $d+4$ in the elements of $\phi$ and a sequence $\Sigma_{\ell,j}$ such that
\begin{align}
    Q= \sum_\ell \gamma_\ell \phi_{\Sigma(\ell,1)}\cdots \phi_{\Sigma(\ell,d+4)} = [\phi_{s_1} \phi_{s_2}\phi_{s_3}\phi_{s_4}, P],
\end{align}
where $\max(|\gamma_\ell|) = \max(|c_j|)$ and $|\{\gamma_\ell\}|\le 2|\{c_j:\{\phi_s\} \bigcap \{\phi_{\sigma_{j,1}}\cdots \phi_{\sigma_{j,d}}\} \not\subseteq \{\emptyset,1\}\}|\le 2 |\{c_j\}|$
If we consider an initial term $P_0$ to be of the form $\phi_{t_1} \phi_{t_2} \phi_{t_3} \phi_{t_4}$ then it is easy to validate that for any $\phi_{s_1}\phi_{s_2}\phi_{s_3}\phi_{s_4}$, $[\phi_{t_1}\phi_{t_2}\phi_{t_3}\phi_{t_4},c\phi_{s_1}\phi_{s_2}\phi_{s_3}\phi_{s_4}]$ is at most an eighth-order polynomial with coefficients at most $c$.  Therefore we have by induction that for any sequence $\chi$

\begin{equation}
    \|[\phi_{\chi_{1,\ell}}\phi_{\chi_{2,\ell}}\phi_{\chi_{3,\ell}}\phi_{\chi_{4,\ell}},[\cdots,[\phi_{\chi_{1,2}}\phi_{\chi_{2,2}}\phi_{\chi_{3,2}}\phi_{\chi_{4,2}},\phi_{\chi_{1,1}}\phi_{\chi_{2,1}}\phi_{\chi_{3,1}}\phi_{\chi_{4,1}}]\cdots]]\|\le 2^{\ell} \max_j\|\phi_j\|^{4\ell} =2^\ell \label{eq:commbd}
\end{equation}
By setting $\ell = N_m + N_{\rm ext} + N_T + N_V$, it therefore follows that the expectation value over indices $\vec{x},\vec{y}$, the norm of the commutators of such terms is with constant probability greater than $2/3$, from the triangle inequality, Eq.~\eqref{eq:unionarg} and the sub-multiplicative property of the operator norm, in
\begin{align}
&O(2^{\ell} \|H_m +H_{L,ext}\|^{N_m +N_{ext}}(N_TN_V)(\mathbb{E}(V))^{N_T}\mathbb{E}(W)^{N_V})\nonumber\\
&\subseteq
    2^{N_m+N_{ext}+ N_T +N_V} N_T N_V ( m+e\max_{\vec{x}}\|A^{ex}(x)\|)^{N_m+N_{ext}} (\max(\kappa_1,\kappa_2)e^2/L)^{N_T+O(1)} \nonumber\\
    &\qquad\times(\max(K_1,K_2,K_3,K_4)/(n_s^{1/3} L) )^{N_V+O(1)})\nonumber\\
    &\subseteq 3^{N_m+N_{ext}+ N_T +N_V} ( m+e\max_{\vec{x}}\|A^{ex}(x)\|)^{N_m+N_{ext}} (\max(\kappa_1,\kappa_2)e^2/L)^{N_T+O(1)}\nonumber\\
    &\qquad\times(\max(K_1,K_2,K_3,K_4)/(n_s^{1/3} L) )^{N_V+O(1)}) \label{eq:commAvg}
\end{align}

The maximum number of non-zero commutators that can arise in the commutator series can be bounded using the following argument.  Let us consider the simplest non-trivial commutator which is of the form $[a_{x}^\dagger a_{y},a_{u}^\dagger a_v]$.  There are clearly $O(n_s^4)$ possible combinations.  However, unless the sets $\{x,y\}$ and $\{u,v\}$ have a non-empty intersection the commutator is zero.  Therefore there are actually $O(n_s^3)$ rather than $O(n_s^4)$ possible non-zero commutators of this form.  Iterating this, it is clear that there are $O(n_s^4) $ non-zero commutators of the form $[a_{s}^\dagger a_{u},[a_{x}^\dagger a_{y},a_u^\dagger a_v]]$.  In general it follows by induction that for the $k$-fold nested commutator (if $k\in O(1)$) that there are at most $n_s^{1+k}$ such terms. 

The situation is even more constrained with terms that arise from the external potential as well as the mass.  Such terms consist of creation and annihilation operators that only act on one fermion site (and $4$ potential components).  Therefore for each such term introduced the site must match one of the other terms in the commutator product otherwise the commutator will be zero.  Thus the number of non-zero commutators in a $k$-fold nested commutator series, where $k\in O(1)$, is also in $O(1)$.

Thus combining these two observations we see that the total number of non-zero commutators that can be formed from a general product is in
\begin{equation}
    O(n_s^{1+N_V+N_T})\, .\label{eq:commnumber}
\end{equation}
Thus the sum over all terms formed by these commutators is simply the number of commutators multiplied by the expectation value of the coefficients.  We can use the expression in Eq.~\eqref{eq:commAvg} to estimate the sum over all commutators by simply multiplying the mean by the number of commutator terms in~\eqref{eq:commnumber}.  This yields

\begin{equation}
    O\left(n_s 3^{N_m+N_{ext} + N_T +N_v} \left( m+e\max_{x}\|A^{ex}(x)\|\right)^{N_m+N_{ext}} \left( \frac{e^2n_s}{L} \right)^{N_T+O(1)} \left(\frac{n_s^{4/3}}{L} \right)^{N_V+O(1)}\right)
\end{equation}
Next let $\Lambda:= \max\left(m+e\max_{x}\|A^{ex}(x)\|,\frac{e^2n_s}{L}, \frac{n_s^{4/3}}{L}\right)$.  The sum of over all commutators with $N_m+N_{ext}+N_T+N_v=\ell$ is then in
\begin{equation}
    n_s\Lambda^{\ell+O(1)} e^{O(\ell)}
\end{equation}
Therefore the size of any commutator that arises from the expansion of the error in the Trotter-Suzuki formula, using time-step $t$, is in
\begin{equation}
    n_s(\Lambda t)^\ell e^{O(\ell)}
\end{equation}
It then follows from Theorem 6 of~\cite{childs2021theory} that the Trotter-Suzuki error for a $p^{\rm th}$ order formula is therefore in
\begin{equation}
    \Delta_{TS}(p) \in n_s(\Lambda t)^{p+1} e^{O(p)}.
\end{equation}
Thus we can choose $t$ such that the error, $\Delta_{TS}$ is at most $\epsilon_{TS}$ for
\begin{equation}
    t\in O\left(\frac{1}{\Lambda}\left(\frac{\epsilon_{TS}}{n_s \Lambda^{O(1)}}\right)^{1/(p+1)} \right)
\end{equation}

Next note that each product formula of order $p$ consists of a product of $5^{p/2-1}$ second-order Trotter formulas.  Each such product formula is composed of $O(n_s^2)$ exponentials.  Therefore we have that the number of exponentials in the product formula is~\cite{berry2007efficient}
\begin{equation}
    N_{\exp} = O\left(n_s^25^{p/2 -1} \right) \  .
\end{equation}
Next let $U_{\rm TS}(t)$ be the Trotter-Suzuki formula for $e^{-iHt}$.  If we apply phase estimation to $U_{\rm TS}(t)$ the eigenvalues returned are, with high probability, those of $Ht$.  Therefore if we wish to estimate the eigenvalues of $H$ within error $O(\epsilon_{\rm TS})$ the phase estimation protocol requires $O(1/(\epsilon_{\rm TS} t))$ repetitions. We choose the error in the phase estimation protocol to be in $\Theta(\epsilon_{TS})$ because we want to ensure that $\epsilon_{PE} + \epsilon_{TS} \le \epsilon$ where $\epsilon $ is our total error tolerance.  This can be attained by choosing $\epsilon_{PE} \in \Theta(\epsilon_{TS})$ as we do here.

The cost of performing phase estimation and estimating the energy within error $\epsilon=\Theta(\epsilon_{TS} t)$ and probability of failure $\delta<1/3$ is in
\begin{equation}
    O\left(\frac{N_{\exp}}{\epsilon_{\rm TS} t} \right)\subseteq O\left(\frac{n_s^{2+1/(p+1)}\Lambda^{1+o(1)} 5^{p/2}}{\epsilon_{\rm TS}^{1+1/(p+1)}} \right) \ .
\end{equation}
In practice, however, we only guarantee that $\|e^{-iHt} - U_{\rm TS}(t) \| \le \epsilon_{\rm TS}$ and need to ensure that the errors in the eigenvalues of the unitaries are comparable.  The necessary result follows from Theorem 6.3.2 of~\cite{horn2012matrix} and using this result and the fact that unitary matrices are unitarily diagonalizable that if $\lambda_x(\cdot)$ is the $x^{\rm th}$ eigenvalue of a matrix then for any $x$ there exists a $y$ such that
\begin{equation}
    |\lambda_x(e^{-iHt}) - \lambda_y(U_{\rm TS}(t))| \le \|e^{-iHt} - U_{\rm TS}(t)\| \le \epsilon_{TS} \ .\label{eq:eigenbd}
\end{equation}
Next, choosing the error from this step such that $\epsilon_{\rm TS} \in \Theta( \epsilon)$ we then have after optimizing over the value of $p$ as per~\cite{berry2007efficient} the number of exponentials needed for the simulation is in
\begin{equation}
    O\left(\frac{N_{\exp}}{\epsilon t} \right)\subseteq\left(\frac{\Lambda^{1+o(1)} n_s^{2+o(1)}}{\epsilon^{1+o(1)} } \right) \ .\label{eq:expscale}
\end{equation}

Gate complexity estimates then easily follow from Eq.~\eqref{eq:expscale}.  The exponential that requires the most $T$-gates to simulate is given in Fig.~\ref{fig:odd_pauli}.  It consists of $8$ Pauli operations.  From Box 4.1 of Nielsen and Chuang~\cite{nielsen2002quantum}, it suffices to synthesize each rotation within error $\epsilon / N_{\exp}$.  Using an optimal synthesis method, such as~\cite{kliuchnikov2012fast,ross2014optimal} this can be achieved using $O(\log(N_{\exp} / \epsilon))$ $T$-gates.  Therefore the number of $T$-gates needed for the simulation is in
\begin{equation}
    N_T \in \left(\frac{\Lambda n_s^{2}}{\epsilon} \right)^{1+o(1)}.\label{eq:TgatesLattice}
\end{equation}
\end{proof}

A key assumption in eQED is that the mass energy of the electrons dominates the momentum contributions.  This is necessary because the derivation of eQED truncates the path integral expansion of the propagator at second order.  The case that most closely resembles the canonical case in the electronic structure literature is where $n_s^{1/3} / L \ll m \ll n_s^{4/3}/L$~\cite{babbush2018low,berry2019qubitization,lee2020even}.  In this non-relativistic case considered in the electronic structure literature, the mass energy of the electron is not included and so the Trotter error is dominated by the momentum of the terms in the Hamiltonian and the number of $T$ gates needed for the simulation scales as $N_T\in (n_s^{10/3}/L\epsilon)^{1+o(1)}$.  If we consider the thermodynamic limit where $L \in \Theta (n_s^{1/3})$, we then have that $N_T \in (n_s^3/\epsilon)^{1+o(1)}$.  This result is comparable to some of the earlier results for simulations of electronic structure in local-bases~\cite{babbush2018low}, but does not precisely match these bounds because of the use of the SLAC kinetic operator, which is much less local than the corresponding kinetic operator used in planewave-dual simulations~\cite{babbush2018low,childs2021theory}.

\subsection{Cost Estimates for Momentum Basis Simulations using Trotter}
The calculation of the norm of the nested commutators for the momentum space Hamiltonian are needed to estimate the Trotterization error for the momentum space simulation.  Fortunately, these nested commutators are much easier to evaluate than their position space brethren because of the lack of summation over auxiliary indices in the definition of the interaction and constraint terms in the Hamiltonian.

\begin{theorem}\label{thm:momCost}
Let $H_p$ be the momentum space effective QED Hamiltonian of Eq.~\eqref{eq:Hrel} and Eq.~\eqref{eq:Hpext} in three spatial dimensions in a cavity of volume $L^3$ with electrons of mass $m$ and charge $e\in O(1)$ and external vector potential $A^{ex}$ such that for any momentum mode within the reciprocal lattice, $|E_p -m|\in o(1)$. Then, there exists a quantum algorithm that when provided a state $\ket{\psi}$ such that $H_p\ket{\psi} = E\ket{\psi}$, the energy value $E$ can be estimated within error $\epsilon$ for any $\epsilon>0$ and failure probability at most $1/3$ using a number of $T$-gates that is in
$$
\left(\frac{n_s^3 \Lambda_p }{ \epsilon} \right)^{1+o(1)} ,
$$
where $\Lambda_p = O\left(n_s\left(\frac{m}{n_s} + \frac{e^2n_s}{L} + e|A^{ex}| \right) \right).$
\end{theorem}
\begin{proof}
The Hamiltonian is the sum of three terms, the kinetic energy term, the electron-electron interaction term and the external potential term.  First let us consider the kinetic term, which is trivial in a momentum basis
\begin{equation}
    H= \sum_{\sigma, \nu} C_{\sigma,\nu} (a^\dagger_{\sigma,\nu} a_{\sigma,\nu}+ b^\dagger_{\sigma,\nu} b_{\sigma,\nu})
\end{equation}
where 
\begin{equation}
    C_{\sigma,\nu} := {E_\nu } = \sqrt{m^2 + \frac{4\pi^2}{L^2} |\vec{\nu}|^2 }\in \Theta(m)
\end{equation}

The two-body interactions are much more complicated in momentum representation.  For example, the fermion-fermion interaction can be written as
\begin{equation}
     \sum_{p,q,r,\sigma_1,\sigma_2,\sigma_3,\sigma_4} e^2D_{p,q,r,\sigma_1,\sigma_2,\sigma_3,\sigma_4} a^{\dagger}_{(p+q-r),\sigma_4} a^{\dagger}_{r,\sigma_3} a_{q,\sigma_2} a_{p,\sigma_1},
\end{equation}
where 
\begin{equation}
    |D_{p,q,r,\sigma_1,\sigma_2,\sigma_3,\sigma_4}| \le \frac{\mathcal{M}}{4L^3 \sqrt{E_{p+q-r} E_r E_p E_q}} \in \Theta\left(\frac{\mathcal{M}}{L^3 m^2} \right) \ . \label{eq:Dbd}
\end{equation}
Here for convenience we take $\mathcal{M}$ an upper bound on the values of the coefficients in Eq.~\eqref{eq:M1}, \eqref{eq:M2}, \eqref{eq:M3}, \eqref{eq:M4} and \eqref{eq:M5}.  By doing so we make the result of~\eqref{eq:Dbd} hold for all the two body interactions in the problem.  First, we see that (in units where $e=1$)
\begin{equation}
\begin{split}
    \mathcal{M} &\in O\left(\max_{p,q}\left(\frac{\max_{j}\{\|u_j(p)\|^4,\|v_j(p)\|^4\} }{|(E_p - E_q)^2 - \frac{4\pi^2}{L^2}|\vec{p} - \vec{q}|^2|} \right)\right) \\
    &\subseteq O\left(\max_p\left(\frac{ (E_p +m)^2}{\min_q |(E_p - E_q)^2 - |\vec{p} - \vec{q}|^2/L^2|} \right)\right) \\
    &\subseteq O\left({m^2L^2}\right)
\end{split}
\end{equation}
Therefore,
\begin{equation}
    |D_{p,q,r,\sigma_1,\sigma_2,\sigma_3,\sigma_4}| \in O\left(\frac{1}{L} \right).
\end{equation}
The exact same scaling holds by inspection for every two body term in the momentum space Hamiltonian.

The external potential (in momentum space) is given by Eq.~\eqref{eq:Hpext}. 
The Hamiltonian in this case can be chosen (in units where $e=1$) to be
\begin{equation}
    H_{p,ext} = \sum_{\sigma_1,\sigma_2,p,q} E_{\sigma_1,\sigma_2,p,q} a^\dagger_p a_q
\end{equation}
where 
\begin{equation}
    |E_{\sigma_1,\sigma_2,p,q}| \in O\left(  \frac{\max_p |u_p|^2 e\max |A^{ex} |}{m} \right) \subseteq O\left(\frac{\max_p(E_p +m) e\max|A^{ex}|}{m} \right) \subseteq O(e|A^{ex}|).
\end{equation}

Next let us consider a Lie-Polynomial of kinetic, mass, interaction and external potential terms consisting of $N_m, N_{\rm ext}, N_T, N_V$ mass, external potential, kinetic and two-body interaction terms.  As noted above, the indices for each mass term must match the indices of existing terms in the polynomial; whereas all other terms must match at least one term.  Therefore the total number of valid commutators that can be present is in
\begin{equation}
    O(n_s^{1+N_{\rm ext} + 2N_V}).
\end{equation}
Next, let us assume that $N_{\rm ext} + N_T +N_V = \ell$.  We can then see that the sum of all nested commutators of order $\ell$ is in

\begin{align}
  &O\left(2^\ell (\max_{\sigma,\nu} |C_{\sigma,\nu}|^{N_T} \max_{\substack{\mathstrut p,q,r,\\ \sigma_1,\sigma_2,\sigma_3,\sigma_4 }}|D_{p,q,r,\sigma_1,\sigma_2,\sigma_3,\sigma_4}|^{N_v} +\max_{p,q,\sigma_1,\sigma_2} |E_{p,q,\sigma_1,\sigma_2}|^{N_{ext}} n_s^{1+N_{ext}+2N_V}) \right)  \nonumber\\
  & \subseteq O\left((2n_s)^\ell \left(\frac{(\max_{\sigma,\nu} |C_{\sigma,\nu}/n_s|}{n_s}\right)^{N_T} \max_{\substack{\mathstrut p,q,r,\\ \sigma_1,\sigma_2,\sigma_3,\sigma_4}}(e^2|D_{p,q,r,\sigma_1,\sigma_2,\sigma_3,\sigma_4}|)^{N_v} \left(n_s{\max_{p,q,\sigma_1,\sigma_2} |E_{p,q,\sigma_1,\sigma_2}|}\right)^{N_{ext}}  \right)
\end{align}
Next if we let 
\begin{align}
    \Lambda_p :=& \max\left\{\frac{\max_{\sigma,\nu} |C_{\sigma,\nu}/n_s|}{n_s}, \max_{\substack{\mathstrut p,q,r,\\\sigma_1,\sigma_2,\sigma_3,\sigma_4}}e^2|D_{p,q,r,\sigma_1,\sigma_2,\sigma_3,\sigma_4}|,\,\, n_s\, {\max_{p,q,\sigma_1,\sigma_2} |E_{p,q,\sigma_1,\sigma_2}|}\right\}\nonumber\\
    &\in O\left(n_s\left(\frac{m}{n_s} + \frac{e^2n_s}{L} + e\max|A^{ex}| \right) \right)
\end{align}
We then have from Theorem 6 of~\cite{childs2021theory} that the error in the $p^{\rm th}$-order Trotter-Suzuki formula is
\begin{equation}
    \Delta_{TS}(p) \in n_s(\Lambda_p t)^{p+1} e^{O(p)}.
\end{equation}
Thus if we wish to have $\Delta_{TS}(p) \le \epsilon_{TS}$ then it suffices to choose
\begin{equation}
    t\in \Theta\left( \frac{1}{\Lambda_p}\left(\frac{\epsilon_{TS}}{n_s}\right)^{1/(p+1)} \right)
\end{equation}
We can then invoke Eq.~\eqref{eq:eigenbd} to show that this corresponds to a systematic error of at most $\epsilon_{TS}$ in the eigenvalues of $e^{-iHt}$ that arises from the Trotter-Suzuki approximation.  Let us define the correct eigenphase that we would see from phase estimation to be $E t$ and the approximate phase $\widetilde{E}_{TS} t$.  This means we can use phase estimation on the Trotter-Suzuki approximation to learn the eigenphase $\widetilde{E}_{TS} t$ within error $\epsilon_{TS} t$ and probability of failure less than $1/3$ using $O(1/\epsilon_{TS} t)$ applications of the Trotter-Suzuki formula~\cite{higgins2007entanglement}.  Thus the total number of operator exponentials that need to be invoked in the simulation in order to learn the $\widetilde{E} t$ within error $O(\epsilon_{TS} t)$ is in
\begin{equation}
    O\left(\frac{N_{\exp}}{\epsilon_{\rm TS} t} \right)\subseteq O\left(\frac{n_s^{3+1/(p+1)}\Lambda_p 5^{p/2}}{\epsilon_{\rm TS}^{1+1/(p+1)}} \right),\label{eq:pNexp}
\end{equation}
Provided that $t(\|H\|+\epsilon_{TS})\le\pi$, we can unambiguously infer $E$ from this result by taking
\begin{equation} \label{eq:energy_phase}
    E = \frac{\widetilde{E} t}{t} + O(\epsilon_{TS}),
\end{equation}
which implies that this estimate also suffices to provide $E$ within error $O(\epsilon_{TS})$ with probability at least $2/3$.  

The final step involves following the reasoning laid out in~\cite{berry2007efficient,childs2021theory} to choose $p$ to minimize the number of operator exponentials needed to achieve error $\epsilon\ge \epsilon_{\rm TS}$.  This corresponds to taking $p\in O(\sqrt{\log(n_s t/\epsilon)})$, which when substituted into Eq.~\eqref{eq:pNexp} leads to a number of operator exponentials that scales as $(\frac{n_s^3 \Lambda_p }{ \epsilon})^{1+o(1)}$, where $(
\cdot)^{o(1)}$ is used to refer to factors that are at most sub-polynomial (but not necessarily poly-logarithmic).

We then see from our circuit constructions that each operator exponential requires a number of $T$ gates that scales as $O(\log(n_s \Lambda t /\epsilon))$ thus the number of $T$-gates required by the simulation obeys
\begin{equation}
    N_T \in  \left(\frac{n_s^3 \Lambda_p }{ \epsilon} \right)^{1+o(1)}.
\end{equation}
\end{proof}

The asymptotics of the simulation complexity in momentum space are interesting for a number of reasons.  First, let us consider the case where the two-body interaction term dominates the Trotter-Suzuki error.  This occurs when $\frac{n_s^{1/3}}{L}\ll m\ll \frac{n_s^2}{L}$.  Note that we require that the lower bound hold in order to justify the assumptions in Theorem~\ref{thm:momCost} as well as to ensure that we remain in the situation where effective QED, rather than full QED, is appropriate.  In the thermodynamic limit, we take $L \propto n_s^{1/3}$ and therefore have that
\begin{equation}
     N_T^{therm} \in \left(\frac{n_s^{4+2/3}}{\epsilon } \right)^{1+o(1)}\label{eq:plimits}
\end{equation}
The continuum limit, unfortunately, is not naturally defined without making further promises on the system.  This is because in the continuum limit we need to take $L\in o(n_s^{1/3})$, which  leads to momentum modes where the kinetic contribution to the energy dominates the mass energy.  Such modes invalidate our assumptions and so  effective QED cannot be considered valid in the continuum limit for finite mass electrons without imposing restrictions on the input state. Note that these issues arise for both the position and momentum space formulations of eQED. Taking the continuum limit is also complicated by the issue of renormalization. To ensure the electron mass takes the correct value and the potential between two electrons has the correct $1/r$ dependence, the electron mass and charge in the Hamiltonian must be varied as a function of the lattice spacing. For the continuum limit of QED, this leads to the electron charge being forced to zero as the lattice spacing goes to zero. This is known as triviality and is due to QED likely not being a valid interacting field theory when defined without a cutoff \cite{triviality1,triviality2,triviality3,triviality4}.

\subsection{Cost Estimates for Qubitization}\label{sec:qubitization}

In recent years, qubitization has emerged as an alternative to Trotter-Suzuki simulations on fault tolerant hardware~\cite{childs2010relationship,low2019hamiltonian,gilyen2019quantum,babbush2018encoding,von2020quantum,lee2020even}.  Unlike Trotter-Suzuki methods, qubitization is known to saturate lower bounds on the query complexity for quantum simulation.  Further, it is much simpler to provide tight bounds for the complexity of simulation using qubitization~\cite{childs2018toward}.  However, qubitization is not space optimal and further cannot directly exploit locality of the Hamiltonian to reduce the costs of simulation unlike Trotter-Suzuki methods.  For these reasons, qubitization does not supplant Trotter-Suzuki methods but rather provide us with a new set of tools that can perform favorably to Trotter-Suzuki methods under certain circumstances.

The central idea of qubitization is that a walk operator, $W\in \mathbb{C}^{(N+M)\times (N+M)}$, can be constructed for any Hamiltonian such that if $H=\sum_j \lambda_j H_j \in\mathbb{C}^{N\times N}$ for unitary $H_j$ then for every eigenvector $\ket{\psi}$ of $H$ and any integer $q$, $W^q \ket{\psi} \ket{0}^M$ is a vector within a two dimensional subspace spanned by the non-orthogonal vectors $\ket{\psi} \ket{0}^M$ and $W\ket{\psi} \ket{0}^M$.  Further let $\lambda = \sum_{j=1}^m |\lambda_j|$.  With these assumptions in place, if we define $\ket{\psi}^\perp$ to be the orthogonal component of $W\ket{\psi}\ket{0}^M$, then within the basis $\ket{\psi}\ket{0}^M$ and $\ket{\psi^\perp}$, the walk operator restricted to this two-dimensional subspace then takes the form
\begin{equation}
    ({\ketbra{\psi}{\psi}\otimes \ketbra{0}{0}+\ketbra{\psi^\perp}{\psi^\perp}})W({\ketbra{\psi}{\psi}\otimes \ketbra{0}{0}+\ketbra{\psi^\perp}{\psi^\perp}})= \begin{bmatrix} \frac{\bra{\psi}H\ket{\psi}}{\lambda} & -\sqrt{1 - \frac{\bra{\psi}H\ket{\psi}^2}{\lambda^2}}\\
    \sqrt{1 - \frac{\bra{\psi}H\ket{\psi}^2}{\lambda^2}} & \frac{\bra{\psi}H\ket{\psi}}{\lambda}
    \end{bmatrix}
\end{equation}
which is isospectral to the rotation $e^{-iY\cos^{-1}(\bra{\psi} H \ket{\psi}/\lambda )}$.  

This shows that if the eigenvalues of $H$ can be estimated, given knowledge of $\lambda$, we can then construct an estimator $\hat{E}$ for the energy, from an estimator for the phase $\hat{\phi}$ using~\cite{babbush2018encoding,poulin2018quantum}
\begin{equation}
    \hat{E} = \lambda\cos(\hat{\phi}).
\end{equation}
Thus the energy can be estimated within error $\epsilon$ using $O(\lambda/\epsilon)$ applications of $W$~\cite{babbush2018encoding}.  
The position space normalization can be found by examining the Jordan-Wigner representation of each of the terms individually.  Specifically, we have that for each $\vec{x}$, $a^\dagger_{\vec{x}}$ is expressed as a sum of $2$ Pauli-operators in the Jordan-Wigner representation.  Since Pauli operators are unitary we can compute the asymptotic scaling of any term in the Hamiltonian by treating the fermionic operators as if they were unitary because only constant factors are introduced by expanding the Jordan-Wigner representation.  We further have that $\lambda = \lambda_m + \lambda_{ext} + \lambda_{int}+\lambda_{SLAC}$, which are the contributions to the normalization terms from the mass, external vector potential and the kinetic operator.  These expressions are straightforward to compute
\begin{align}
    H_m &= \sum_{\vec{x}} m a_{\vec{x}}^\dagger \gamma^0 a_{\vec{x}} \Rightarrow \lambda_m \in O(m n_s)\nonumber\\
    H_{L,ext} &= \sum_{\vec{x}} e a_{\vec{x}}^\dagger \gamma^0 \gamma^\mu A_\mu^{ex}(\vec{x}) a_{\vec{x}} \Rightarrow \lambda_{ext} \in O\left(n_s e\max_{\vec{x}}|A^{ex}(\vec{x})| \right)\nonumber\\
    H_{SLAC}&= \frac{2\pi}{n_s^{1/3} L }\sum_{\vec{x},\vec{y},\vec{p}} e^{-i 2\pi n_s^{-1/3} \vec{p}\cdot(\vec{x} -\vec{y})} a^\dagger_{\vec{x}}\gamma^0 \gamma^j p_j a_{\vec{x}} \Rightarrow \lambda_{SLAC} \in O \left(\frac{n_s^{5/3}}{L} \right)\nonumber\\ 
    H_{int} &= \sum_{\vec{x} \ne \vec{y}} \sum_{\mu,\nu} g_{\mu\mu} \frac{n_s^{1/3} e^2}{8\pi L |\vec{x} -\vec{y}|}(a_{\vec{x}}^\dagger \gamma^0 \gamma^\mu a_{\vec{x}})(a_{\vec{y}}^\dagger \gamma^0 \gamma^\nu a_{\vec{y}}) \Rightarrow \lambda_{int} \in O\left( \frac{n_s^2 e^2}{L}\right),\label{eq:lambdaEqns}
\end{align}
where the last expression in Eq.~\eqref{eq:lambdaEqns} follows from the average over position of $1/|\vec{x}-\vec{y}|$ given in~\eqref{eq:Vavg}.
We therefore have that the normalization constant in position space eQED is
\begin{equation}
    \lambda_{pos} \in O\left(mn_s +n_s e\max_{\vec{x}} |A^{ex}(\vec{x})| +(n_s^{-1/3} + e^2)\frac{n_s^2}{L})  \right)
\end{equation}

It is straightforward to compute the values of $\lambda$ for the momentum space formalism for eQED under worst-case assumptions about the functional form of the external vector potential $A^{ex}(x)$:
\begin{align}
    \lambda_{mom}&\in O\left(m n_s+ \frac{n_s^3}{L} + n_s^2 e \max_{\vec{x}} |A^{ex}(\vec{x})|   \right) 
\end{align}

More specifically, the walk operator is constructed from a pair of unitary operations $\mathtt{SELECT}$ and $\mathtt{PREPARE}$.  For simplicity let us assume without loss of generality that $\lambda_j\ge 0$ (any signs or phases can be absorbed into the $H_j$). The operator $\mathtt{PREPARE}$ is defined to prepare an initial state
\begin{equation}
    \mathtt{PREPARE} \ket{0} = \frac{1}{\sqrt{m}} \sum_j \sqrt{\frac{\lambda_j}{\lambda}}\ket{j}.\label{eq:prep}
\end{equation}
Note the operation of $\mathtt{PREPARE}$ on states other than $\ket{0}$ is not specified here because any unitary matrix that satisfies Eq.~\eqref{eq:prep} can be used to construct the walk operator $W$.

The action of select is similarly defined via
\begin{equation}
    \mathtt{SELECT} \ket{j} \ket{\psi} = \ket{j} H_j \ket{\psi}.
\end{equation}
The walk operator $W$ is then defined to be
\begin{equation}
    W:= (1 - 2 \mathtt{PREPARE} \ketbra{0}{0} \mathtt{PREPARE}^\dagger) \mathtt{SELECT}.
\end{equation}

It is then clear from this exposition that the cost of performing the quantum simulation depends directly on two quantities: the normalization constant $\lambda$ and the costs of performing $\mathtt{SELECT}$ and $\mathtt{PREPARE}$.  The costs, however, depend sensitively on the construction used for these two operations and the circuit constructions for the two operations are complicated relative to those used in Trotter-Suzuki simulations.

For simplicity, we will adapt the construction of~\cite{berry2019qubitization,babbush2018encoding} which was derived for simulations of non-relativistic chemistry in an arbitrary basis to the relativistic case considered here.  The prepare circuit is implemented using a memory access model known as a QROM, which can be thought of as an oracle replacement that uses a lookup table to store each of the unique amplitudes in the state $\mathtt{PREPARE}\ket{0}$.  If there are $K$ such amplitudes then the cost of preparing the state within error $\epsilon$ is in $O(K + \log(1/\epsilon))$ using the approach outlined in Section 3.D of~\cite{babbush2018encoding}.   

The number of unique coefficients in the position space Hamiltonian, $K_{pos}$, is substantially lower than the number of terms in the Hamiltonian.  While the number of terms in the postion space Hamiltonian is in $O(n_s^2)$ only $O(n_s)$ of these can take unique values.  This can easily be seen from Eq.~\eqref{eq:lambdaEqns} wherein $H_m$ only takes $O(1)$ values, $H_{L,ext}$ takes at most $O(n_s)$ values (assuming each $A^{ex}(\vec{x})$ is unique).  $H_{SLAC}$ contains $O(n_s^{2/3})$ unique exponentials of the form $e^{-i2\pi n_s^{-1/3}\vec{p} \cdot (\vec{x}-\vec{y})}$ and $O(n_s^{1/3})$ values of $p_j$.  Thus the total number of distinct amplitudes for $H_{SLAC}$ is at most in $O(n_s)$ as well.  Finally the fermion-fermion interaction consists of only $O(n_s^{1/3})$ distinct values and hence 
\begin{equation}
    K_{pos}\in O(n_s).
\end{equation}

The number of unique coefficients in the momentum space Hamiltonian is much more challenging to analyze and such simple patterns in the magnitudes of the coefficients do not naturally reveal themselves.  As such, we use the trivial bound of 
\begin{equation}
    K_{mom} \in O(n_s^3).
\end{equation}
It is likely, however, that by refactoring the momentum space Hamiltonian using techniques analogous to~\cite{motta2018low} that a substantial reduction in $K_{mom}$ may be attainable.

The last piece that needs to be considered is the implementation of $\mathtt{SELECT}$.  The approach that we use again mirrors the presentation in Fig. 9 of~\cite{babbush2018encoding}.  The strategy we take is to decompose the fermionic operators into Majorana operators of the form $X\otimes Z\otimes \cdots\otimes Z$ and $Y\otimes Z\otimes \cdots Z$.  At most four Majorana operators are needed in both position and momentum space and thus the cost of implementing the select operator is $O(1)$ times the cost of applying the selected Majorana operator.  The construction in~\cite{Babbush18_110501} allows such Majorana operators to be selected in time $O(n_s)$ and therefore the cost of the select circuit is in $O(n_s)$ in both the position and momentum bases.  Therefore, with this construction the cost of state preparation dominates the cost of the select circuit.

The number of $T$-gates needed in the qubitized simulation is therefore the product of the number of applications of $W$ needed by phase estimation and the sum of the number of $T$-gates needed by the prepare and select circuits.  This means that the complexity for the position space simulation under the assumption that $e^2 \gg n_s^{-1/3}$ is
\begin{equation}
    N_{T,pos} \in \widetilde{O}\left(\frac{\lambda_{pos}( K_{pos}+n_s)}{\epsilon} \right)\subseteq \widetilde{O}\left(\frac{n_s^2(m + e\max_{\vec{x}} |A^{ex}(\vec{x})| + e^2\frac{n_s}{L})}{\epsilon}  \right).
\end{equation}
Similarly, the number of gates needed to perform the momentum space simulation at most scales as
\begin{equation}
    N_{T,mom} \in \widetilde{O}\left(\frac{\lambda_{mom}( K_{mom}+n_s)}{\epsilon} \right)\subseteq \widetilde{O}\left(\frac{n_s^4(m + \frac{n_s^2}{L} + n_s e \max_{\vec{x}} |A^{ex}(\vec{x})|)}{\epsilon}  \right).
\end{equation}
If we assume the thermodynamic limit, then we have that the scaling of qubitization is upper bounded by $\widetilde{O}(\frac{n_s^{2+2/3}}{\epsilon})$ in position space and $\widetilde{O}(\frac{n_s^{5+2/3}}{\epsilon})$ in momentum space.  The scaling of qubitization in position basis is slightly superior to the upper bound on the scaling of Trotterization, $(\frac{n_s^{3}}{\epsilon})^{1+o(1)}$.  In momentum basis, the use of a brute force prepare circuit switches this behavior and leads to worse scaling than the $(n_s^{4+2/3}/\epsilon)^{1+o(1)}$ scaling provided by Trotter formulas.  This bound, however, is likely pessimistic and by taking advantage of symmetries in the Hamiltonian terms it is likely that the number of unique coefficients can be further compressed. A summary of the results presented in this section are given in Table \ref{table:t-gate-summary}.
\begin{table}[ht]
\begin{center}
\begin{tabular}{ |c|c|c|c| }
\hline
 Method & Hamiltonian Basis & T-gate Complexity \\
\hline
Trotter-Suzuki & Position & $O(n_s^{3}/\epsilon)^{1+o(1)}$ \\ 
Trotter-Suzuki & Momentum & $O(n_s^{4+2/3}/\epsilon)^{1+o(1)}$\\ 
Qubitization & Position &$\widetilde{O}(n_s^{2+2/3} /\epsilon)$ \\ 
Qubitization & Momentum & $\widetilde{O}(n_s^{5+2/3} / \epsilon)$ \\ 
\hline
\end{tabular}\caption{The T-gate complexities for both Trotter-Suzuki and Qubitization simulations in the position and momentum based eQED Hamiltonians in the thermodynamic limit.}\label{table:t-gate-summary}
\end{center} 
\end{table} 

A final point of interest is that the performance of Trotter-Suzuki methods in the non-relativistic limit may be superior to qubitization.  Specifically, if we define the non-relativistic limit to be the case where $m\gg n_s^{2}/L$, then the scaling of position space simulation using Trotter-Suzuki methods becomes $(n_s^2m/\epsilon)^{1+o(1)}$.  On the other hand, qubitization's cost scales as $\widetilde{O}(n_s^2 m /\epsilon)$ in this limit.  Thus for cases where relativistic effects are small, but highly accurate simulations are required then the bounds for Trotterization coincide (up to sub-polynomial factors) with those of qubitization.  Further, since the empirical performance of Trotter-Suzuki methods is often much better than the upper bounds~\cite{reiher2017elucidating,childs2018toward} it is natural to suspect that Trotter's performance may be even better than this bound as also noted in the following section.
% -------------------------------------------------------

\section{Rellium Model Analysis}
\subsection{Numerical Evaluation of Momentum Space Commutators}
In this section, we present a numerical study of the momentum space rellium Hamiltonian commutators.  We focus on the momentum space Hamiltonian because said Hamiltonian can be simply constructed, where the spinor interaction terms are computed with planewave integrals. For the following simulations, the Hamiltonian terms and integral coefficients were constructed utilizing the SymPy software package~\cite{sympy_package}, and the cell box size was kept constant at $L = 1$. Each successive model system with different numbers of planewaves $n_s$, were created by modifying the planewave energy cutoff, $E_{cut}$ within this constant box size. Additionally, renormalization terms in the rellium Hamiltonian, Eq.~\eqref{eq:Hrel}, were ignored for simplicity.

In the following result, the expectation value over $i,j,k$ of the nested commutator, also called the second order commutator, $\| [ H_i, [H_j, H_k]] \|$ was computed by randomly sampling Hamiltonian terms using Monte Carlo. The indices $i,j,k$ in this case represent any possible term in the Hamiltonian. A total of 8 different rellium systems were used by defining the planewave energy cutoff values at $E_{cut} = 8, 9, 11, 14, 15, 17, 20, 23, 26$ eV which correspond to the number of planewaves being $N = 24, 72, 104, 128, 224, 320, 584, 808, 1216$ respectively. Each average commutator value was computed by running a Monte Carlo simulation with increasing sample numbers in order to evaluate the limit of the average commutator, and then taking the average of the corresponding Monte Carlo runs. Specifically, 13 simulations were performed for each rellium system equally spaced on the logarithmic scale between a minimum of $10^6$ and a maximum of $10^{10}$ samples inclusively.

\begin{figure}[ht]
	\begin{center}
		\includegraphics[width=3.75in]{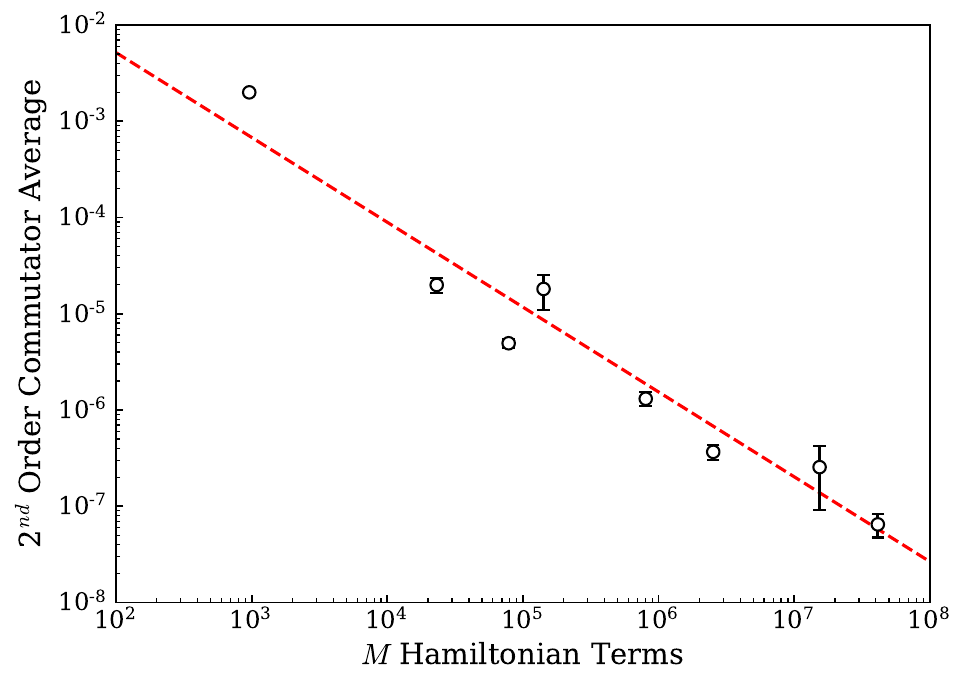}
		\caption{The Monte Carlo sampled 2$^{\text{nd}}$ order commutator average for different $E_{cut}$ values defining different rellium systems with constant box length $L = 1$. Plotted points are the average value across all Monte Carlo runs, and the error bars denote the standard deviation. The dashed red line corresponds to the least-squares power law fit where $f(M) = 0.3 M^{-0.9}$}
		\label{fig:commutator}
	\end{center}
\end{figure}

The results of the Monte Carlo simulations are presented in Fig. \ref{fig:commutator}, where each data point is the average of all runs for each system, including the standard deviation in the error bars. By using a least squares power law fit, we can compare how the number of terms in the Hamiltonian, $M$, affect the average second order commutator. The power law fit function from the commutator average data in Fig. \ref{fig:commutator}, results in $f(M) = 0.3M^{-0.9}$. Using the fact that the number of terms in the Hamiltonian $M$ scales as $O(n_s^3)$ with the number of planewaves, we can na\"ively state that the complexity of the second order commutator would scale as $O(n_s^9)$. However, since there needs to be at least one index in common between each Hamiltonian term, the upper bound is actually $O(n_s^7)$. By multiplying the upper bound for the number of second order commutator terms to evaluate, $n_s^7$, with the fitted function $f(M)$, where $M = n_s^3$, we see that the relationship of this complexity is observed to be $\sim O(n_s^4)$, performing better than the upper bound. The derivation of the estimated T-gate complexity is detailed in the following section.

\subsection{Cost Estimate for QPE}
The most common goal of Hamiltonian simulation in general is to find ground state energies. By using the above rellium model for analysis, we can gain an understanding of the error scaling throughout the Trotter-Suzuki decomposition, quantum phase estimation (QPE), and $T$-gate synthesis. First, in the standard surface code model of fault-tolerant quantum computation, all Clifford gates are trivial in cost, where non-Clifford gates such as the $T$-gate end up dominating the computational cost of an arbitrary circuit. Therefore, we can simply understand the cost of our simulation with the number of $R_z$ gates required, which are typically implemented as a circuit of multiple $T$-gates. 

Second, it is important to note that the problem of finding ground states is generically hard.  Specifically, if we could find a ground state energy in polynomial time then the complexity class $\BQP$ would contain $\QMA$, which is the quantum analog of $\P=\NP$.  For this reason, we strongly suspect that the ground state preparation problem is intractable on quantum computers unless a sufficiently good guess of the ground state can be provided to the phase estimation algorithm.  Although we discuss in the following section strategies that can ameliorate this problem, it is important to note that we do not necessarily know how well these methods will work with a particular instance of a rellium simulation and hence the only thing we can say with confidence is that the estimates contained here will give the cost of sampling from the spectrum of rellium.

In order to eventually compute the ground state energy tolerance within some precision $\epsilon$, we first need to find out the error scaling in the Trotter-Suzuki (TS) decomposition, $\epsilon_{TS}$. Specifically we will focus on the second order decomposition of our unitary time propagator the error in which is given by~\cite{childs2021theory} to be at most

\begin{align}
    \epsilon_{TS} \leq \frac{t^3}{12} \sum_{\gamma} \left| \left|  \left[ \sum_{\alpha} H_{\alpha},  \left[   \sum_{\beta} H_{\beta}, H_{\gamma} \right] \right] \right| \right| + \frac{t^3}{24} \sum_{\gamma} \left| \left|  \left[ H_{\gamma},  \left[   \sum_{\beta} H_{\gamma}, H_{\beta} \right] \right] \right| \right|
    &\le \frac{t^3}{8}\sum_{\gamma,\alpha,\beta}\left| \left|  \left[ H_{\alpha},  \left[   H_{\beta}, H_{\gamma} \right] \right] \right| \right|
\end{align}
%If we require an overall error in the simulation to be $\epsilon$, this yields a resulting error in $\epsilon_{TS}$ in the measured phase, denoted as $\widetilde{E} t$, as described in Eq. \ref{eq:energy_phase} (More description here) 
Using Ref. \cite{Neven18_041015}, the root-mean-square error in the measured phase during phase estimation can be denoted as the following
\begin{equation}
    \Delta \phi \approx \sqrt{\left( \frac{\pi}{2^{n + 1}}\right)^2 + \left(\epsilon_{\text{TS}} + \epsilon_{\text{syn}} + \pi \epsilon_{\text{QFT}} \right)^2}
\end{equation}
where $n$ is the number of ancilla qubits used.
We will now neglect the cost of the quantum Fourier transform as it needs to be done only once and so $\epsilon_{QFT}$ can be taken to be an incredibly small value without altering the time-complexity of the simulation.
\begin{equation}
    \Delta \phi \approx \sqrt{\left( \frac{\pi}{2^{n + 1}}\right)^2 + \left(\epsilon_{\text{TS}} + \epsilon_{\text{syn}} \right)^2}
\end{equation}
We can now set the phase error target to be equivalent to the total RMS error in the energy multiplied by the total time propagation $t$
\begin{equation}
    {\epsilon} = \frac{\Delta \phi}{t} \approx \frac{1}{t}\sqrt{\left( \frac{\pi}{2^{n + 1}}\right)^2 + \left(\epsilon_{\text{TS}} +\epsilon_{\text{syn}}\right)^2}
\end{equation}
For simplicity, we choose $\epsilon_{TS} =\epsilon_{syn} = \pi\sqrt{3}/2^{n+2}$.  With this choice we find
\begin{equation}
    \epsilon \approx \frac{1}{t} \sqrt{\frac{1}{4} \left(\frac{\pi}{2^n} \right)^2 + \frac{3}{4}\left(\frac{\pi}{2^n} \right)^2  }=\frac{\pi}{t2^n} ,
\end{equation}
where $n$ is the number of qubits used in the phase estimation routine. 

Next, using the fitted function from the Monte Carlo simulation given in Figure~\ref{fig:commutator}, we can take 
\begin{equation}
   \sum_{\gamma,\alpha,\beta}\left| \left|  \left[ H_{\alpha},  \left[   H_{\beta}, H_{\gamma} \right] \right] \right| \right| \approx \frac{0.3 n_s^7}{M^{0.9}} = 0.3 n_s^{4.3}\label{eq:empscaling}
\end{equation}
which we can define as 
\begin{equation}
   \chi_H = A n_s^{b}
\end{equation}
where $A = 0.3$ and $b = 4.3$. Note that  if we assumed the worst case commutator scaling that would be predicted from the commutators would be $b=7$ for momentum basis simulations.  This shows that substantial gaps likely exist between the worst case scalings and the actual scaling for eQED, similar to observations that have already been made for quantum simulations of non-relativistic chemistry. 

Therefore, we can substitute the estimate in Eq.~\eqref{eq:empscaling} into $\epsilon_{TS}$ to find
\begin{equation}
    \frac{\pi \sqrt{3}}{2^{n+2}} \leq \frac{t^3 \chi_H}{8} \  .
\end{equation}
Now we find that the correct choice of $t$, relative to these bounds, will satisfy
\begin{equation}
    \sqrt[3]{\frac{\pi \sqrt{3} }{\chi_H 2^{n-1}}} \leq t \  .
\end{equation}
Picking $t$ to saturate the lower bound (which corresponds to the worst-case scenario) we find that $t = \sqrt[3]{\frac{\pi \sqrt{3} }{\chi_H 2^{n-1}}}$, and can then solve for the number of qubits required in the QPE procedure.
\begin{equation}\label{eq:nQPE}
    n = \ceil[\Bigg]{\frac{\log(\frac{\pi^2\chi_H}{{2} \sqrt{3} {\epsilon^3 }})}{2\log(2)}}  \ .
\end{equation}

For a single trotter step, the number of rotations required is based on the total number of terms in the Hamiltonian
\begin{equation}
    N_{terms} = 2n_s + 9n_s^3  \  .
\end{equation}
Using the above formula, the max number of rotations required for a single term being equal to 8, and the fact that the number of exponentials required for the second order TS formula is $N_{exp} = 2N_{terms}$; The total number of rotations per trotter step is 
\begin{equation}
    N_{Rot} = 8 \times 2N_{terms} \le 32n_s + 144n_s^3  \ .
\end{equation}

The number of rotations needed in QPE is $2^n$, therefore the overall number of rotations needed for the simulation is
\begin{equation}
    N_{Rot}^{Sim} = 2^{n} \left( 32n_s + 144n_s^3 \right)
\end{equation}

Using chemical accuracy, $\epsilon = 1.6\text{mHa}$ for our error target, we can then estimate the number of qubits needed for QPE, and ultimately the number of $T$-gates required to obtain the ground state energy eigenvalue within the error tolerance of choice. The number of $T$-gates per rotation can be computed using the scaling from Ref. \citenum{Svore15_080502}, and our chosen error for the $T$-gate synthesis, $\epsilon_{syn}$, where 
\begin{equation}
    N_T = 1.15\log_2(1/\epsilon_{syn}) \times N^{Sim}_{Rot} \  .
\end{equation}
Since the error in the eigenvalue scales at most linearly with the error in the unitary matrix~\cite{horn2012matrix} and since the error in the unitary scales at most linearly with the number of gates comprising the unitary from Box 4.1 of~\cite{nielsen2002quantum} we have that it suffices to take $\epsilon_{syn} = \epsilon N_{Rot}$.  With this assignment $N_T$ becomes 
\begin{equation}\label{eq:nt_qpe}
    N_T = 1.15\log_2(N_{Rot}/\epsilon) \times N^{Sim}_{Rot}
\end{equation}
Using the empirical values for $A$ and $b$ defined above, chemical accuracy $\epsilon$, the number of qubits necessary for QPE in Eq.~\eqref{eq:nQPE}, the relation for $T$-gate count in Eq.~\eqref{eq:nt_qpe}, we can finally estimate the $T$-gate count for the full QPE routine given some number of planewaves for the system, $n_s$. The log-log plot of this relationship is given in Fig. \ref{fig:Tgate}.
\begin{figure}[t]
	\begin{center}
		\includegraphics[width=3.75in]{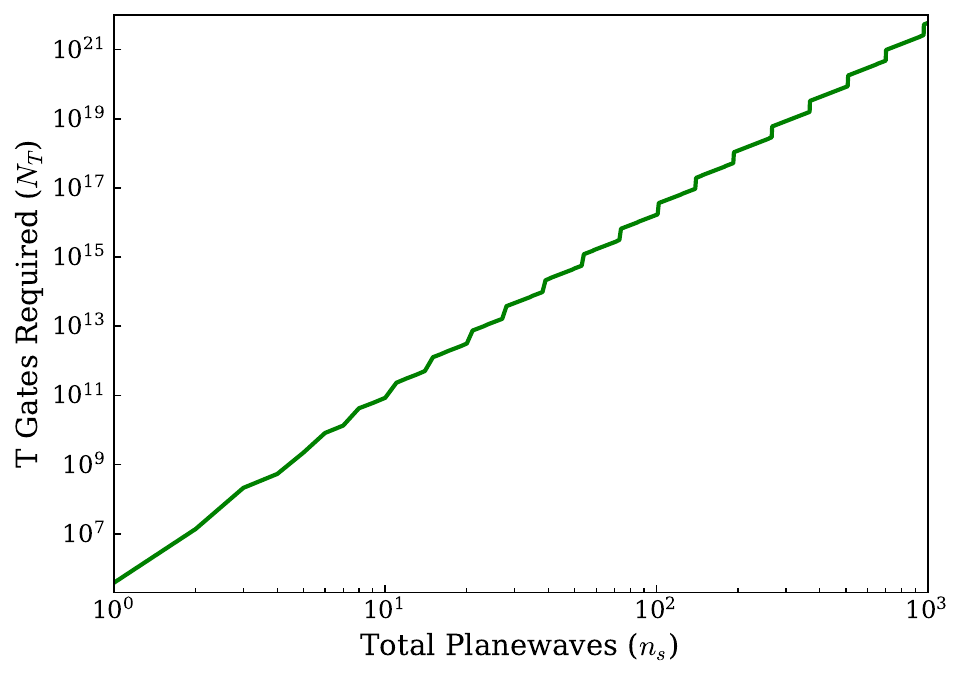}
		\caption{The estimated total number of $T$-gates required to sample from the spectrum of the rellium Hamiltonian for a box of length $L=1$ within an error tolerance of $\epsilon = 1.6\text{mHa}$ corresponding to chemical accuracy as a function of the number of planewaves in the system, $n_s$. }
		\label{fig:Tgate}
	\end{center}
\end{figure}
The full equation for this relationship is
\begin{equation}
%    N_T = 586250n_s^{5.3}\label{eq:finalscale}
    N_T =\ceil[\Bigg]{1.15 \times 2^n \left(32n_s + 144n_s^{3} \right) \log_2 \left(\frac{\left(32n_s + 144n_s^{3} \right)}{\pi \sqrt{3} / 2^{n+2}} \right)}\in \widetilde{O}\left(\frac{n_s^3\sqrt{\chi_H}}{\epsilon^{3/2}}\right) = \widetilde{O}\left(\frac{n_s^{5.2}}{\epsilon^{3/2}}\right)\label{eq:empSim}
\end{equation}
for an error of $\epsilon = 1.6\text{mHa}$, and $n$ is determined from Eq.~\eqref{eq:nQPE}.

In contrast to Eq.~\eqref{eq:empSim}, the costs for such a planewave simulation in the constant $L$ and $\epsilon$ limit are given by Theorem~\ref{thm:momCost} to be in $n_s^{5+o(1)}$ for adaptively chosen high-order Trotter-Suzuki formulas.  We find empirically that for the second-order Trotter-Suzuki formula the number of gates needed to reach $1.6\text{mHa}$ is in $\widetilde{O}(n_s^{5.2})$.  This suggests that, the empirical performance of this simulation is comparable to what we expect from our prior theoretical analysis; however, it is worth considering that this analysis still does rely on crude triangle-inequality based estimates that disregard cancellation between terms in the expansion of the error operator, and that further studies may be needed to determine the impact of neglecting such cancellations.

Finally, the number of non-Clifford operations needed to perform a classically challenging simulation using $20$ plane waves (comprising $80$ logical qubits) is projected by our results to be on the order of $10^{13}$ non-Clifford operations.  In contrast, the best known results for simulating jellium using Trotter methods are on the order of $10^9$ non-Clifford operations for systems of $27$ spin orbitals ($54$ qubits).  The gulfs between these two estimates suggest that further optimization may be needed to allow eQED to reach the same levels of performance that we can reach for non-relativistic electronic structure calculations, however, the gulfs between the two are not so large as to suspect that such simulations will be infeasible once subjected to the same optimizations that lowered the costs of simulation for challenge problems in chemistry from $10^{14}$ non-Clifford gates~\cite{reiher2014relativistic} to on the order of $10^{9}$ non-Clifford gates~\cite{lee2020even,von2020quantum}

% -------------------------------------------------------

\section{State Preparation}

In order to efficiently simulate any Hamiltonian on a quantum register, it is necessary to start with a high quality wavefunction ans\"atz that has a large overlap with the true wavefunction. For the jellium or rellium Hamiltonian, the degeneracy of the ground state is a non-trivial problem. Typically this is referred to as ``strong correlation" where many different electronic configurations are entangled and contribute to the ground state wavefunction. This is in contrast to the familiar Hartee-Fock ground state, where a single electronic configuration \emph{is} the wavefunction. Ground state electronic systems where the Hartree-Fock configuration is not dominant are sometimes called multi-reference (MR) ground states, where using the Hartree-Fock reference has a small overlap with the true ground state wavefunction due to other configurations being just as important. The multi-reference nature of electronic systems is also present in many chemical systems, typically large conjugated carbon systems, and multi-metal centered molecules due to the abundance of low lying spin states.

A commonly known classical method for computing multi-reference wavefunctions in quantum chemistry is called multi-reference configuration interaction (MRCI).~\cite{Shephard12_108, Shavitt81_91} In short, this method captures static correlation by first expanding the Hartree-Fock reference into a complete active space (CAS) of all configurations included within a truncated orbital and particle space, typically centered around the occupied and unoccupied valence orbitals. Next, a number of additional configurations are added for capturing correlation effects beyond the CAS space, sometimes referred to as ``dynamic correlation." A common version of MRCI, is MRCI singles and doubles (MRCISD), where additional determinants are added to the CAS wavefunction, including single and double particle excitations. Additionally, the singly and doubly excited particle and hole space spans a larger space then just the initial CAS determinants, and are commonly referred as restricted active space 1 (RAS1) for the additional occupied orbitals, and subsequently RAS3 for the virtual orbitals. The single and double excitations span across all three spaces. As an aside, RAS2 sometimes refers to the original CAS space in the literature.

The number of determinants required to build an MRCI ans\"atz can be defined by the chosen active spaces, RAS1, CAS, and RAS3. Following the notation in Ref. \citenum{Li20_2975} where we only consider singly occupied spin-orbitals that follow the Jordan-Wigner mapping.
\begin{align}
N_{\text{det}}^{MRCI} &= \sum_{i_h = 0}^{m_h}\sum_{i_e = 0}^{m_e} {N_{RAS1} \choose i_h} {N_{CAS} \choose n_e - N_{RAS1} + i_h - i_e} {N_{RAS3} \choose i_e} 
\end{align}
where $N_{RAS1}$, $N_{CAS}$, and $N_{RAS3}$ are the number of orbitals in each active space respectively, $n_e$ is the total number of electrons in all of the correlated spaces, $i_h$ is the hole index, $i_e$ is the particle index, $m_h$ and $m_e$ are the number of holes in RAS1 and the particles in RAS3 respectively. For the MRCISD ans\"{a}tz the number of determinants then scales exponentially with $N_{CAS}$, but quadratically in both $N_{RAS1}$ and $N_{RAS3}$. Additionally for MRCISD, $m_h = m_e = 2$.

For the purpose of state-preparation on a quantum computer, we can pre-compute the MRCISD wave function classically, and then use the coefficients for the determinants to initialize a better wavefunction that hopefully has a much higher overlap with the true ground state. A convenient method to prepare the initial MRCISD state is to use a Givens rotation on the creation and annihilation operators, defined by a general Slater determinant generated by the classical MRCISD wavefunction. \cite{Boixo18_044036, babbush2018low} The general Slater determinant can be defined as 
\begin{equation}
| \Phi \rangle  = \hat{d}^\dagger_1 \cdots \hat{d}^\dagger_{N_f} | 0 \rangle, \qquad \hat{d}^\dagger_i = \sum_j^{n_s} Q_{ij} \hat{c}^\dagger_{j}
\end{equation}
where $\hat{c}^\dagger$ is the arbitrary particle type creation operator in the computational basis, $\hat{d}^\dagger$ is the rotated creation operator for the new basis based off of fractional particle occupation, $N_f$ is the number of fermions, and $n_s$ is the number of orbitals in the chosen basis. The $N_f \times n_s$ matrix $Q$ rotates the original creation/annihilation operators into the rotated basis based off the choice of initial wavefunction. Therefore the rows of $Q$ correspond to single particle wavefunctions that are linear combinations of the original orbital basis. In general, the rotated Slater determinant can be generated by a unitary rotation of a simple computational basis state generated by the original creation operators
\begin{equation}
| \Phi \rangle  = U \hat{c}^\dagger_1 \cdots \hat{c}^\dagger_{N_f} | 0 \rangle, \qquad \hat{d}^\dagger_i = U \hat{c}^\dagger_i  U^\dagger
\end{equation}
Following Ref. \citenum{Boixo18_044036} we see that this unitary transformation can be represented by sequences of 2-qubit rotations, also known as Givens rotations where
\begin{equation}
U = G_{N_G} \cdots G_2 G_1
\end{equation}
where $N_G$ is the total number of Givens rotations, and the Givens rotation matrix between spin-orbitals $p,q$ is defined as
\begin{equation}
G(\theta, \varphi) = 
\begin{pmatrix} 
\cos \theta & -e^{i \varphi} \sin \theta \\
\sin\theta & e^{i \varphi} \cos \theta \\
\end{pmatrix}
\end{equation}
where the angles $\theta$ and $\varphi$ can be solved classically by diagonalizing the $Q$ matrix, where $Q = V^{\dagger} D U$. The number of Givens rotations required to perform the basis transform is
\begin{equation}
N_G = \left(n_s - N_f \right) N_f \in O(n_s^2).
\end{equation}
Each Givens rotation, $G(\theta,\varphi)$, can be implemented using a rotation of the form $e^{-i(e^{-i\varphi}a_p^\dagger a_q + e^{i\varphi}a_q^\dagger a_p)\theta}$. As discussed, each exponential can be simulated using $O(\log(1/\epsilon))$ $T$-gates.  The complexity of performing this state preparation is therefore in $O(n_s^2\log(n_s^2/\epsilon))$, which is sub-dominant to the cost of simulation in momentum basis which is in $\widetilde{O}(n_s^{4+2/3}/\epsilon)$ in the thermodynamic limit.  For this reason, we neglect the cost of the state preparation in all of our previous analysis.

 Additionally, for typical electronic systems of interest, we only care about linear combinations of electronic single particle functions and the positronic block is trivially occupied. This means that for a system of $N_{elec}$ number of electrons our matrix $Q$ will already be diagonal in the positronic space, meaning simply that $N_f = N_{elec}$.

\section{Planewave Cutoff Estimates for Heavy Atoms}

In this section we will provide heuristic arguments that estimate the number of planewaves needed to solve a realistic relativistic problem in the momentum basis.  This is important because the cost of both the MRCISD ans\"{a}tz as well as the simulation of the dynamics can be non-trivial.  As a target problem, we focus on the simulation of atomic gold.  This is chosen because relativistic effects are needed in order to even qualitatively understand the spectrum of gold and thus such a simulation is arguably the first logical benchmark problem to consider after simulation of the relativistic free electron gas (rellium). In the estimates below we use the atomic unit convention where $\hbar = m_e = e = 1$.

For the momentum space planewave Hamiltonian, we estimate that a single all-electron gold atom will require at least 31 million planewaves, which is obviously beyond the reach of quantum computers in the foreseeable future. The number of planewaves $N_{PW}$ needed for an arbitrary system is defined by the cell volume $L^3$ and energy cutoff.
\begin{equation}
    N_{PW} = \frac{L^3}{2\pi^2} E_{cut}^{3/2}
\end{equation}
To find the lowest possible cutoff energy, we can calculate the highest possible kinetic energy of an electron in an atomic potential with nuclear charge $Z$, which will be in the 1$s$ orbital. The kinetic energy of the 1$s$ electron can be estimated to be the following, using the hydrogenic Dirac equation
\begin{equation}
%    E_{kin} = \frac{Z^2}{2n^2}
    E = c^2 \left[ \frac{1}{\sqrt{1+\frac{Z^2 \alpha^2}{\left(n-(j+1/2)+\sqrt{(j+1/2)^2-Z^2\alpha^2} \right)^2}}} -1 \right]
\end{equation}
where $\alpha$ is the fine structure constant, $n = 1$ is the principle quantum number and for the ground state $s$ orbital, $j = \frac{1}{2}$. By plugging in $Z = 79$ and $L = 2r_{A}$ where $r_{A}$ is the atomic radius, using 3.14$a_0$ for gold , and finally setting $E_{cut} = E_{kin}$ we estimate that the all electron gold atom to require at least $3.08\times 10^7$ planewaves and in turn roughly $1.23\times 10^8$ logical qubits. This means that based on the number of sites in the reciprocal lattice and the cost of the simulation, enforcing chemical accuracy $(1.6 \text{ mHa})$, will require roughly on the order of $10^{38}$ non-Clifford operations.  Our analysis therefore suggests that such simulations will likely be out of reach for any quantum computer in the next few decades and beyond, since the number of planewaves for a single atom all electron system is expected to increase as $O(L^3Z^3)$. We expect similar conclusions to hold for other heavy element atomic systems as well. 

This is not too surprising since planewave simulations in general need a large number of basis functions to properly describe atomic core orbitals. The most obvious remedy to this issue is to switch to a different basis set.  In particular, Gaussian orbitals model the nuclear cusp condition much better than planewaves and so the number of Gaussians needed to describe the system to within chemical accuracy can be substantially lower. This makes them often a more natural choice.

The opposite approach would be to instead of investigating eQED in second quantization to look at it instead in first quantization using an appropriately anti-symmetrized wave function.  Within such a framework, the number of qubits needed to store the atomic configuration can be exponentially smaller.  The prefactors, however, make such applications outside of the reach of existing or near-future quantum computers.

\section{Conclusion}
In this work we have presented how to simulate the eQED Hamiltonian on a quantum computer. Specifically, we presented an analysis of both the position basis using a cubic lattice, and the momentum basis planewave formulations of the Hamiltonians. From this analysis we find that for the position basis, the number of $T$-gates required for simulating the Hamiltonian scales as $\left( \frac{\Lambda n_s^2}{\epsilon}\right)^{1+o(1)}$ where $\Lambda = \max\left(m+e\max_{x}\|A^{ex}(x)\|,\frac{e^2n_s}{L}, \frac{n_s^{4/3}}{L}\right)$. For the momentum basis, the number of $T$-gates required scales as $\left(\frac{ \Lambda_p n_s^3}{\epsilon} \right)^{1+o(1)}$ where $\Lambda_p \in O\left(n_s\left(\frac{m}{n_s} + \frac{n_s}{L} + e|A^{ex}_\mu| \right) \right)$.  This shows that the the ground state energy can be computed in polynomial time given that a copy of the ground state is provided to a quantum computer, which suggests that the problem of deciding whether there exists an eigenstate with energy less than a threshold, is contained within the complexity class \QMA \, for effective QED.  This is relevant because it provides further vindication that relativistic effects do not invalidate the Quantum Complexity-Theoretic Church-Turing Thesis.

Further, we investigated the cost of using quantum phase estimation to estimate the ground state energy of the relativistic jellium (rellium) model as the simplest momentum based eQED Hamiltonian. Specifically we computed the number of $T$-gates needed empirically for quantum phase estimation using the second order Trotter-Suzuki decomposition, and using Monte Carlo sampling of different rellium Hamiltonians by increasing the cutoff energy for the system. For this routine, we find that for a constant box size, $L = 1$, the number of $T$-gates needed to estimate the ground state energy eigenvalue within an error of chemical accuracy $\epsilon = 1.6\text{mHa}$ is on the order of $10^{16}$ $T$-gates for a classically intractable problem involving $100$ planewaves ($400$ qubits) or $10^{13}$ $T$-gates for a classically challenging problem with $20$ plane waves $(80)$  qubits.  These costs, while substantial, suggest that by further optimizing our simulation algorithm that the costs of quantum simulation may be reduced to reasonable levels.

While this work has explored how eQED can be simulated in general, the momentum space and lattice formulations of the Hamiltonian are not necessarily the most pertinent for all applications to physics and chemistry. Specifically, the focus of including QED corrections to relativistic effects in molecular systems is most prominent for heavy elements at the bottom of the periodic table which have large potential wells from the nuclear charge. However, QED corrections to the energies and properties of light elements can be important in certain situations as well. Future work will focus on adapting this method for simulating eQED on quantum computers to other more convenient basis sets for chemistry, such as the well known Gaussian basis sets that can more compactly model the electronic wavefunction at the nuclear cusp. 

\section*{Acknowledgements}
We would like to thank Hongbin Liu and Evelyn Goldfield for helpful discussions and feedback.  NW was funded by a grant from Google Quantum AI, the NSF funded Challenge Institute for Quantum Computation, the Pacific Northwest
National Laboratory LDRD program and the “Embedding Quantum Computing into Many-Body Frameworks for
Strongly Correlated Molecular and Materials Systems” project, funded by the US Department of Energy (DOE). AC was supported in part by the INT's U.S. Department of Energy grant No. DE-FG02- 00ER41132 and in part by the U. S. Department of Energy grant No. DE-SC0019478. XL acknowledges support from the U.S. Department of Energy, Office of Science, Basic Energy Sciences, in the Heavy-Element Chemistry program (Grant No. DE-SC0021100). TS was supported by a fellowship from The Molecular Sciences Software Institute under NSF grant OAC-1547580.

%%% here %%%

\bibliographystyle{unsrtnat}
\bibliography{bibliography}

%\begin{thebibliography}{9}
%\bibitem{examplecitation}
%  Name Surname,
%  \href{https://doi.org/10.22331/
%        idonotexist}{Quantum
%        \textbf{123}, 123456 (1916).}
%
%\bibitem{biblatexsubmittingtothearxiv}
%  StackExchange discussion on \href{http://tex.stackexchange.com/questions/26990/biblatex-submitting-to-the-arxiv}{``Biblatex: submitting to the arXiv'' (2017-01-10)}
%
%\bibitem{arxivpdfoutput}
%  Help article published by the arXiv on \href{https://arxiv.org/help/submit_tex}{``Considerations for TeX Submissions'' (2017-01-10)}
%
%\bibitem{howtogetdoilinksinbibliography}
%  StackExchange discussion on \href{http://tex.stackexchange.com/questions/3802/how-to-get-doi-links-in-bibliography}{``How to get DOI links in bibliography'' (2016-11-18)}
%  
%\bibitem{automaticallyaddingdoifieldstoahandmadebibliography}
%  StackExchange discussion on \href{http://tex.stackexchange.com/questions/6810/automatically-adding-doi-fields-to-a-hand-made-bibliography}{``Automatically adding DOI fields to a hand-made bibliography'' (2016-11-18)}
%\end{thebibliography}

\onecolumn\newpage
\appendix

\section{Momentum Space Hamiltonian}
\label{appendix:momentumpotential}
The interactions in an effective field theory should be chosen to correctly reproduce some physics of the full model. To correctly reproduce the physics of QED, the effective interaction will be chosen to correctly reproduce the QED scattering amplitudes at lowest order in perturbation theory. This means the potential will consist of 4 fermion terms which describe the scattering processes $e^{\pm} e^{\pm} \rightarrow  e^{\pm} e^{\pm}$, $e^{+} e^{-} \rightarrow  e^{+} e^{-}$,  $e^\pm \rightleftharpoons e^\pm e^{+} e^{-}$ and $0 \rightleftharpoons e^{+} e^{-} e^{+} e^{-}$. Note that the $1 \rightarrow 3$ and $0 \rightarrow 4$ scattering amplitudes will always be off-shell, so these scattering processes will not be directly observed, but including them in the Hamiltonian is necessary for scattering amplitudes at higher orders to be correctly reproduced. The necessary scattering amplitudes in the following subsections are computed using Feynman diagrams at leading order.  Further details of these derivations can be found in~\cite{Schwartz}.
\subsection{$e^\pm e^\pm \rightarrow e^\pm e^\pm$ Amplitudes}
The electron scattering amplitude is given by

\begin{equation}
\mathcal{M}_{e^-_{p,\sigma_1}e^-_{p,\sigma_2}}^{e^-_{r,\sigma_3}e^-_{p+q-r,\sigma_4}} =   e^2 \left(\frac{\bar{u}_{\sigma_3}(p_3)\gamma^\mu u_{\sigma_1}(p_1) \bar{u}_{\sigma_4}(p_4)\gamma_\mu u_{\sigma_2}(p_2)}{(E_{p_3} - E_{p_1})^2-(\vec{p_3} - \vec{p_1})^2} - \frac{\bar{u}_{\sigma_4}(p_4)\gamma^\mu u_{\sigma_1}(p_1) \bar{u}_{\sigma_3}(p_3)\gamma_\mu u_{\sigma_2}(p_2)}{(E_{p_4} - E_{p_1})^2-(\vec{p_4} - \vec{p_1})^2}\right) \ .   
\label{eq:M1}    
\end{equation}
The positron scattering amplitude takes a similar form and it is given by

\begin{equation}
\mathcal{M}_{e^+_{p,\sigma_1}e^+_{q,\sigma_2}}^{e^+_{r,\sigma_3} e^+_{p+q-r,\sigma_4}} =  e^2 \left(\frac{\bar{v}_{\sigma_1}(p_1)\gamma^\mu v_{\sigma_3}(p_3) \bar{v}_{\sigma_2}(p_2)\gamma_\mu v_{\sigma_4}(p_4)}{(E_{p_3} - E_{p_1})^2-(\vec{p_3} - \vec{p_1})^2} - \frac{\bar{v}_{\sigma_1}(p_1)\gamma^\mu v_{\sigma_4}(p_4) \bar{v}_{\sigma_2}(p_2)\gamma_\mu v_{\sigma_3}(p_3)}{(E_{p_4} - E_{p_1})^2-(\vec{p_4} - \vec{p_1})^2}\right) \ .    \label{eq:M2}
\end{equation}

\subsection{$e^+ e^- \rightarrow e^+ e^-$ Amplitude}
The electron positron scattering amplitude is given by
\begin{equation}
\mathcal{M}_{e^-_{p,\sigma_1}e^+_{q,\sigma_2}}^{e^-_{r,\sigma_3} e^+_{p+q-r,\sigma_4}} = e^2 \left(\frac{\bar{v}_{\sigma_2}(q_1) \gamma^\mu u_{\sigma_1}(p_1) \bar{u}_{\sigma_3}(p_2) \gamma_\mu v_{\sigma_4}(q_2)}{(E_{p_1} + E_{p_2})^2-(\vec{p_1} + \vec{p_2})^2} + \frac{\bar{u}_{\sigma_3}(p_2) \gamma^\mu u_{\sigma_1}(p_1) \bar{v}_{\sigma_3}(q_1) \gamma_\mu v_{\sigma_4}(q_2)}{(E_{p_1} - E_{p_2})^2-(\vec{p_1} - \vec{p_2})^2} \right) \ .    \label{eq:M3}
\end{equation}

\subsection{$e^\pm  \rightarrow e^+ e^- e^\pm$ Amplitude}
The $e^- \rightarrow e^- e^+ e^-$ scattering amplitude is given by
\begin{equation}
\mathcal{M}_{e^-_{p,\sigma_1}}^{e^+_{q,\sigma_2} e^-_{p_1,\sigma_3} e^-_{p-q-p_1,\sigma_4}} = e^2 \left(\frac{\bar{u}_{\sigma_3}(p_1) \gamma^\mu u_{\sigma_1}(p) \bar{u}_{\sigma_4}(p_2) \gamma_\mu v_{\sigma_2}(q)}{(E_{p} - E_{p_1})^2-(\vec{p} - \vec{p_1})^2} - \frac{\bar{u}_{\sigma_4}(p_2) \gamma^\mu u_{\sigma_1}(p) \bar{u}_{\sigma_3}(p_1) \gamma_\mu v_{\sigma_2}(q)}{(E_{p} - E_{p_2})^2-(\vec{p} - \vec{p_2})^2} \right) \ .   \label{eq:M4} 
\end{equation}
The $e^+ \rightarrow e^+ e^+ e^-$ scattering amplitude is given by
\begin{equation}
\mathcal{M}_{e^+_{p,\sigma_1}}^{e^-_{q,\sigma_2} e^+_{p_1,\sigma_3} e^+_{p-q-p_1,\sigma_4}} = e^2 \left(\frac{\bar{u}_{\sigma_2}(q) \gamma^\mu v_{\sigma_4}(p_2) \bar{v}_{\sigma_1}(p) \gamma_\mu v_{\sigma_3}(p_1)}{(E_{p} - E_{p_1})^2-(\vec{p} - \vec{p_1})^2} - \frac{\bar{u}_{\sigma_2}(q) \gamma^\mu v_{\sigma_3}(p_3) \bar{v}_{\sigma_1}(p) \gamma_\mu v_{\sigma_4}(p_2)}{(E_{p} - E_{p_2})^2-(\vec{p} - \vec{p_2})^2} \right) .    \label{eq:M5}
\end{equation}

\subsection{Na\"ive Vacuum Coupling}
The amplitude describing the coupling of the na\"ive vacuum to states with nonzero electron and positron number is given by

\begin{equation}
\mathcal{M}_{0}^{e^-_{p,\sigma_1} e^+_{q,\sigma_2} e^-_{r,\sigma_3} e^+_{-p-q-r,\sigma_4}} = \frac{\bar{u}_{\sigma_1}(p) \gamma^\mu v_{\sigma_2}(q) \bar{u}_{\sigma_3}(r) \gamma_\mu v_{\sigma_4}(-p-q-r)}{(E_{p} + E_{q})^2-(\vec{p} + \vec{q})^2} \ .
\end{equation}

\section{Diagonalization of Interaction Terms}
\label{appendix:diagonalization}
In order to derive our simulation circuits for the imaginary terms in the eQED Hamiltonian we need to show explicit simulation circuits for the mutually commuting Pauli operators of the form $YXXX,XYXX,\ldots,XYYY$.
There are eight possible combinations of which four cases need to be considered.  To see why this is, let us consider the cases $XYYY,YXYY,YYXY,YYYX$.  The GHZ preparation circuit, $G$, has a symmetry in that the circuit is invariant under swaps of qubits $2,3$ and $4$.  Thus if we utilize this permutational symmetry, the only cases that need to be considered are $XYYY$ and $YXYY$ as the other two are equivalent to $YXYY$ under exchange of the last three qubits.  The exact same argument holds true for $XXXY,\ldots,YXXX$ and so only four cases need to be considered to understand how the modified GHZ preparation circuit $G$ diagonalizes such terms.

\begin{table}[h!]
    \centering
    \begin{tabular}{|c|c||c|c|}
    Original&Equivalent&Original&Equivalent\\
    \hline
         $G^\dagger (XYYY)G$ & $-ZZZZ$ & $G^\dagger (YXXX) G$ & $Z111$  \\
         $G^\dagger (YYYX) G$ & $-ZZZ1$ & $G^\dagger (XXXY)G$ & $Z11Z$  \\
         $G^\dagger (YYXY) G$ & $-ZZ1Z$ & $G^\dagger (XXYX)G$ & $Z1Z1$  \\
         $G^\dagger (YXYY) G$ & $-Z1ZZ$ & $G^\dagger (XYXX)G$ & $ZZ11$  \\
    \end{tabular}
    \caption{Summary of diagonalization of Hamiltonian terms using the GHZ preparation circuit $G$}
    \label{tab:my_label}
\end{table}

The proof that the circuit $G$ introduced in section \ref{section:intcircuits} diagonalizes all the relevant Pauli operators will make use of the $X$ and $Z$ error propagation identity shown in Fig. \ref{fig:z_error_circuit}.
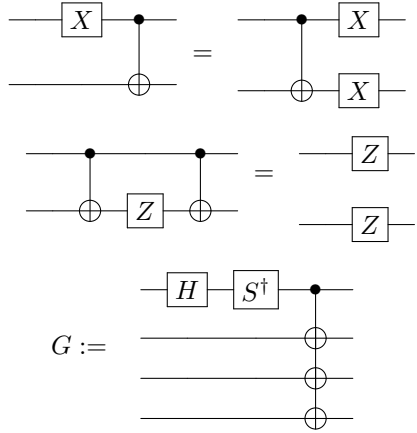
\begin{figure}[t!]
\centering
\subfloat{
\Qcircuit @C=1em @R=.7em {
& \qw & \gate{X} & \ctrl{2} & \qw&& \\
&   &           &           &   &=&\\
& \qw &  \qw     &  \targ   & \qw&& \\ 
}}
\subfloat{
\Qcircuit @C=1em @R=.7em {
& \qw & \ctrl{2} & \gate{X} & \qw& \\
&       &       &           &   &\\
& \qw & \targ    & \gate{X} & \qw& \\ 
}}   

\subfloat{
\Qcircuit @C=1em @R=.7em {
& \qw & \ctrl{2} &  \qw     & \ctrl{2} & \qw && \\
&   &           &           &           &   &=&\\
& \qw &  \targ   &  \gate{Z}&  \targ   & \qw&& \\ 
}}
\subfloat{
\Qcircuit @C=1em @R=.7em {
& \qw &  \gate{Z} & \qw \\
&   &           &\\
& \qw &  \gate{Z} & \qw \\ 
}}\\
\subfloat{\Qcircuit @C=1em @R=.7em {&\\&\\&\\G:=&\\}}\subfloat{
    \Qcircuit @C=1em @R=.7em {
&&\gate{H} &\gate{S^\dagger} &\ctrl{3}  &\qw       \\
&&\qw      &\qw              &\targ    &\qw       \\ 
&&\qw      &\qw              &\targ     &\qw       \\
&&\qw      &\qw              &\targ        &\qw       \\
}}

    \caption{$X$ and $Z$ Error Propagation Identities and definition of the GHZ transformation circuit $G$}
    \label{fig:z_error_circuit}
\end{figure}

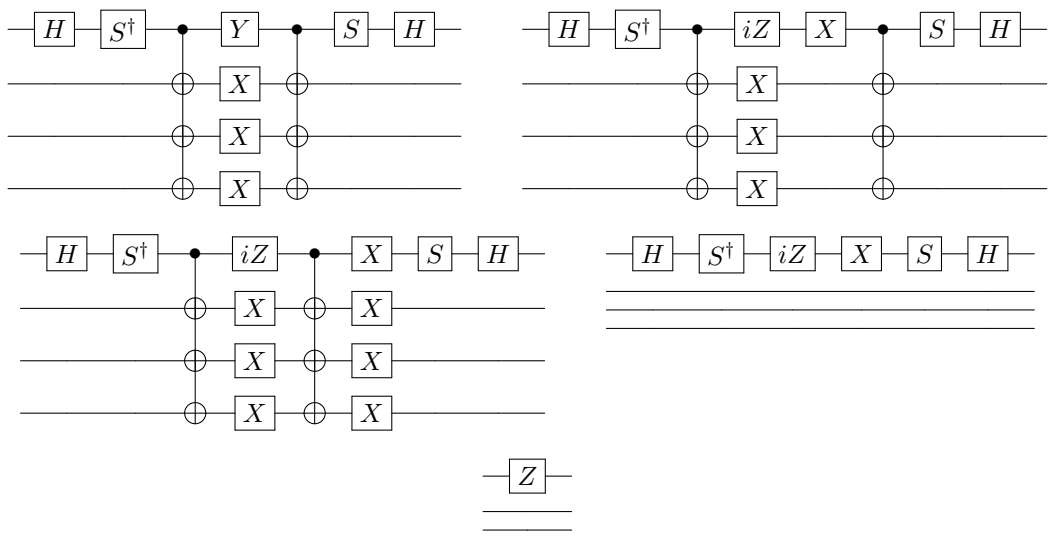
\begin{figure}[t!]
    \centering
\subfloat{
    \Qcircuit @C=1em @R=.7em {
&\gate{H} &\gate{S^\dagger} &\ctrl{3} &\gate{Y} &\ctrl{3} &\gate{S} &\gate{H} &\qw       \\
&\qw      &\qw              &\targ    &\gate{X} &\targ    &\qw      &\qw      &\qw       \\ 
&\qw      &\qw              &\targ    &\gate{X} &\targ    &\qw      &\qw      &\qw       \\
&\qw      &\qw              &\targ    &\gate{X} &\targ    &\qw      &\qw      &\qw       \\
}}
\qquad
%\subfloat{$ = $ (Using $Y=iXZ$)}
\subfloat{
    \Qcircuit @C=1em @R=.7em {
&\gate{H} &\gate{S^\dagger} &\ctrl{3} &\gate{i Z} & \gate{X} &\ctrl{3} &\gate{S} &\gate{H} &\qw       \\
&\qw      &\qw              &\targ    &\gate{X}   &\qw       &\targ    &\qw      &\qw      &\qw       \\ 
&\qw      &\qw              &\targ    &\gate{X}   &\qw       &\targ    &\qw      &\qw      &\qw       \\
&\qw      &\qw              &\targ    &\gate{X}   &\qw       &\targ    &\qw      &\qw      &\qw       \\
}}
\\
\subfloat{
    \Qcircuit @C=1em @R=.7em {
&\gate{H} &\gate{S^\dagger} &\ctrl{3} &\gate{i Z} &\ctrl{3} &\gate{X} &\gate{S} &\gate{H} &\qw        \\
&\qw      &\qw              &\targ    &\gate{X}   &\targ    &\gate{X} &\qw      &\qw      &\qw        \\ 
&\qw      &\qw              &\targ    &\gate{X}   &\targ    &\gate{X} &\qw      &\qw      &\qw        \\
&\qw      &\qw              &\targ    &\gate{X}   &\targ    &\gate{X} &\qw      &\qw      &\qw        \\
}}
\qquad
\subfloat{
    \Qcircuit @C=1em @R=.7em {
&\gate{H} &\gate{S^\dagger} &\gate{i Z}  &\gate{X} &\gate{S} &\gate{H} &\qw       \\
&\qw      &\qw              &\qw         &\qw      &\qw      &\qw      &\qw       \\ 
&\qw      &\qw              &\qw         &\qw      &\qw      &\qw      &\qw       \\
&\qw      &\qw              &\qw         &\qw      &\qw      &\qw      &\qw       \\
}}
\\
\subfloat{
    \Qcircuit @C=1em @R=.7em {
& \gate{Z} &\qw       \\
&\qw       &\qw       \\ 
&\qw       &\qw       \\
&\qw       &\qw       \\
}}
    \caption{Diagonalization of $YXXX$}
    \label{fig:yxxx_diag}
\end{figure}
In the first line of Fig. \ref{fig:yxxx_diag}, the identity $Y=iXZ$ has been used to replace the $Y$ gate acting on the first qubit. The $X$ error identity was then used to move the $X$ gate acting on the first qubit past the CNOT gate. Since a $Z$ gate acting on the control of a CNOT commutes with the CNOT and an $X$ gate acting on the target of a CNOT commutes with the CNOT, the $Z$ gate acting on the first qubit can be moved past the CNOTs and the $X$ gates acting on the lower qubits can be moved past the CNOTs. This results in a cancellation of all of the CNOTs and $X$'s acting on the three lower qubits. Calculating $G^\dagger (YXXX) G$ is reduced to computing the product of the single qubit gates in the second to last diagram of Fig. \ref{fig:yxxx_diag} which concludes the proof that $G^\dagger (YXXX) G = Z111$. The same techniques are applied in the following diagrams to diagonalize the remaining relevant Pauli operators.

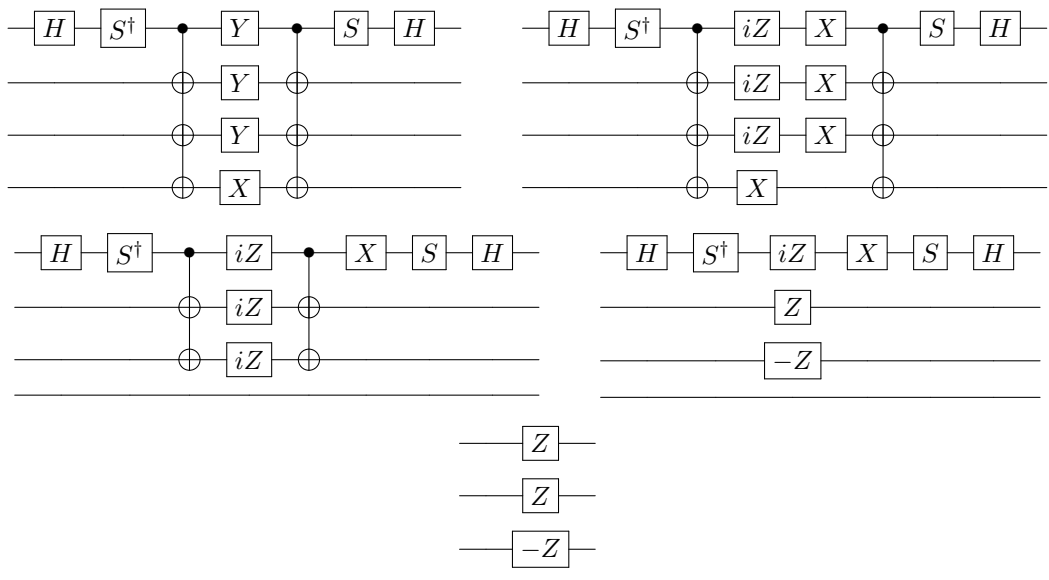
\begin{figure}[t!]
    \centering
\subfloat{
    \Qcircuit @C=1em @R=.7em {
&\gate{H} &\gate{S^\dagger} &\ctrl{3} &\gate{Y} &\ctrl{3} &\gate{S} &\gate{H} &\qw       \\
&\qw      &\qw              &\targ    &\gate{Y} &\targ    &\qw      &\qw      &\qw       \\ 
&\qw      &\qw              &\targ    &\gate{Y} &\targ    &\qw      &\qw      &\qw       \\
&\qw      &\qw              &\targ    &\gate{X} &\targ    &\qw      &\qw      &\qw       \\
}}
%\subfloat{$ = $ (Using $Y=iXZ$)}    
\qquad
\subfloat{
    \Qcircuit @C=1em @R=.7em {
&\gate{H} &\gate{S^\dagger} &\ctrl{3} &\gate{i Z} &\gate{X}  &\ctrl{3} &\gate{S} &\gate{H} &\qw       \\
&\qw      &\qw              &\targ    &\gate{i Z} &\gate{X}  &\targ    &\qw      &\qw      &\qw       \\ 
&\qw      &\qw              &\targ    &\gate{i Z} &\gate{X}  &\targ    &\qw      &\qw      &\qw       \\
&\qw      &\qw              &\targ    &\gate{X}   &\qw       &\targ    &\qw      &\qw      &\qw       \\
}}

%\subfloat{$ = $ (Using multiple applications of the $X$ error identity)}
\subfloat{
    \Qcircuit @C=1em @R=.7em {
&\gate{H} &\gate{S^\dagger} &\ctrl{2} &\gate{i Z} &\ctrl{2} &\gate{X} &\gate{S} &\gate{H} &\qw        \\
&\qw      &\qw              &\targ    &\gate{i Z} &\targ    &\qw      &\qw      &\qw      &\qw        \\ 
&\qw      &\qw              &\targ    &\gate{i Z} &\targ    &\qw      &\qw      &\qw      &\qw        \\
&\qw      &\qw              &\qw      &\qw        &\qw      &\qw      &\qw      &\qw      &\qw        \\
}}
%\subfloat{$ = $ (Using multiple applications of the $Z$ error identity)}
\qquad
\subfloat{
    \Qcircuit @C=1em @R=.7em {
&\gate{H} &\gate{S^\dagger} &\gate{i Z} &\gate{X} &\gate{S} &\gate{H} &\qw        \\
&\qw      &\qw              &\gate{Z}        &\qw      &\qw      &\qw      &\qw        \\ 
&\qw      &\qw              &\gate{- Z} &\qw      &\qw      &\qw      &\qw        \\
&\qw      &\qw              &\qw        &\qw      &\qw      &\qw      &\qw        \\
}}
%\subfloat{$ = $}

\subfloat{
    \Qcircuit @C=1em @R=.7em {
&\qw      &\gate{Z} &\qw        \\
&\qw      &\gate{Z}      &\qw        \\ 
&\qw      &\gate{-Z} &\qw        \\
&\qw      &\qw      &\qw        \\
}}

    \caption{Steps involved in showing the diagonalization of $YYYX$ using GHZ transformations.}
    \label{fig:yyyx_diag}
\end{figure}

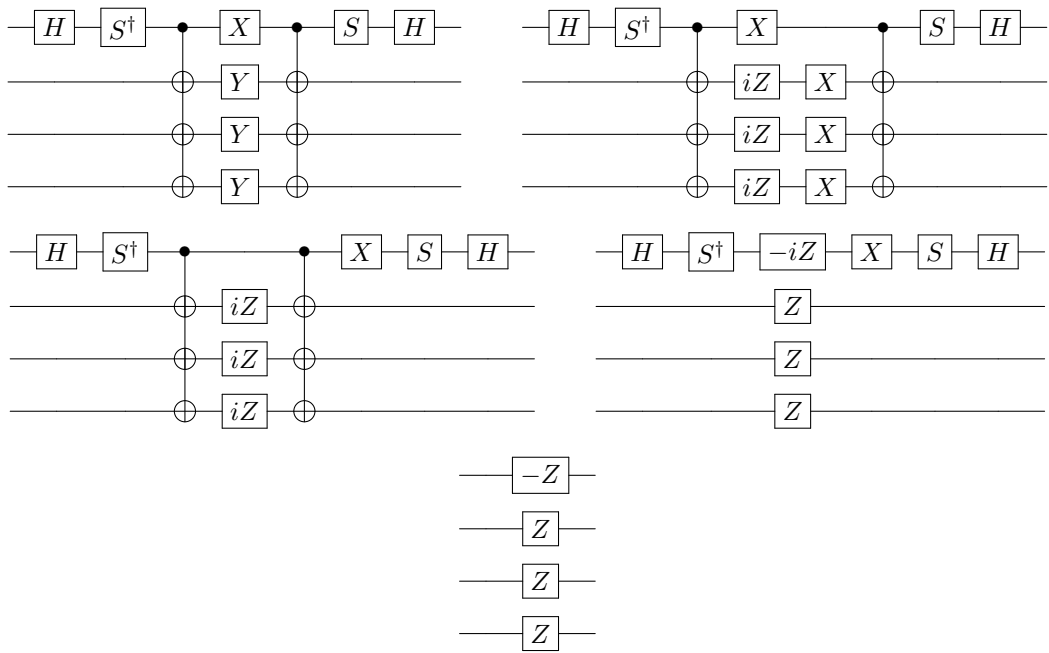
\begin{figure}[t!]
    \centering
\subfloat{
    \Qcircuit @C=1em @R=.7em {
&\gate{H} &\gate{S^\dagger} &\ctrl{3} &\gate{X} &\ctrl{3} &\gate{S} &\gate{H} &\qw       \\
&\qw      &\qw              &\targ    &\gate{Y} &\targ    &\qw      &\qw      &\qw       \\ 
&\qw      &\qw              &\targ    &\gate{Y} &\targ    &\qw      &\qw      &\qw       \\
&\qw      &\qw              &\targ    &\gate{Y} &\targ    &\qw      &\qw      &\qw       \\
}}
\qquad
\subfloat{
    \Qcircuit @C=1em @R=.7em {
&\gate{H} &\gate{S^\dagger} &\ctrl{3} &\gate{X}   &\qw       &\ctrl{3} &\gate{S} &\gate{H} &\qw       \\
&\qw      &\qw              &\targ    &\gate{i Z} &\gate{X}  &\targ    &\qw      &\qw      &\qw       \\ 
&\qw      &\qw              &\targ    &\gate{i Z} &\gate{X}  &\targ    &\qw      &\qw      &\qw       \\
&\qw      &\qw              &\targ    &\gate{i Z} &\gate{X}  &\targ    &\qw      &\qw      &\qw       \\
}}
\\
\subfloat{
    \Qcircuit @C=1em @R=.7em {
&\gate{H} &\gate{S^\dagger} &\ctrl{3} &\qw        &\ctrl{3} &\gate{X} &\gate{S} &\gate{H} &\qw        \\
&\qw      &\qw              &\targ    &\gate{i Z} &\targ    &\qw      &\qw      &\qw      &\qw        \\ 
&\qw      &\qw              &\targ    &\gate{i Z} &\targ    &\qw      &\qw      &\qw      &\qw        \\
&\qw      &\qw              &\targ    &\gate{i Z} &\targ    &\qw      &\qw      &\qw      &\qw        \\
}}
\qquad
\subfloat{
    \Qcircuit @C=1em @R=.7em {
&\gate{H} &\gate{S^\dagger} &\gate{-iZ}   &\gate{X} &\gate{S} &\gate{H} &\qw        \\
&\qw      &\qw              &\gate{Z} &\qw      &\qw      &\qw      &\qw        \\ 
&\qw      &\qw              &\gate{Z}      &\qw      &\qw      &\qw      &\qw        \\
&\qw      &\qw              &\gate{Z} &\qw      &\qw      &\qw      &\qw        \\
}}
\\
\subfloat{
    \Qcircuit @C=1em @R=.7em {
&\qw      &\gate{-Z}  &\qw        \\
&\qw      &\gate{Z}   &\qw        \\ 
&\qw      &\gate{Z}        &\qw        \\
&\qw      &\gate{Z}   &\qw        \\
}}

    \caption{Diagonalization of $XYYY$}
\end{figure}

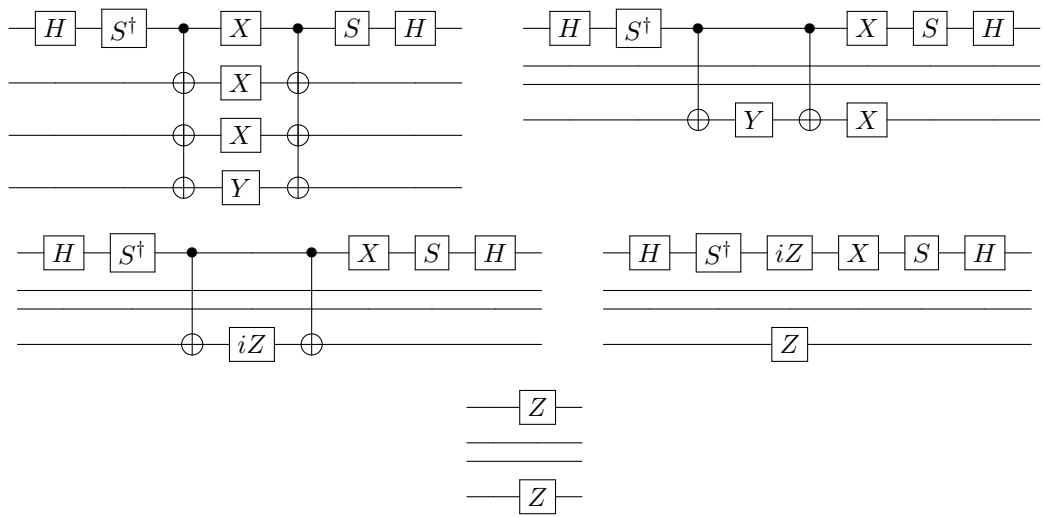
\begin{figure}[t!]
    \centering
    
\subfloat{
    \Qcircuit @C=1em @R=.7em {
&\gate{H} &\gate{S^\dagger} &\ctrl{3} &\gate{X} &\ctrl{3} &\gate{S} &\gate{H} &\qw       \\
&\qw      &\qw              &\targ    &\gate{X} &\targ    &\qw      &\qw      &\qw       \\ 
&\qw      &\qw              &\targ    &\gate{X} &\targ    &\qw      &\qw      &\qw       \\
&\qw      &\qw              &\targ    &\gate{Y} &\targ    &\qw      &\qw      &\qw       \\
}}
\qquad
\subfloat{
    \Qcircuit @C=1em @R=.7em {
&\gate{H} &\gate{S^\dagger} &\ctrl{3} &\qw      &\ctrl{3} &\gate{X} &\gate{S} &\gate{H} &\qw       \\
&\qw      &\qw              &\qw      &\qw      &\qw      &\qw      &\qw      &\qw      &\qw       \\ 
&\qw      &\qw              &\qw      &\qw      &\qw      &\qw      &\qw      &\qw      &\qw       \\
&\qw      &\qw              &\targ    &\gate{Y} &\targ    &\gate{X} &\qw      &\qw      &\qw       \\
}}
\\
\subfloat{
    \Qcircuit @C=1em @R=.7em {
&\gate{H} &\gate{S^\dagger} &\ctrl{3} &\qw        &\ctrl{3} &\gate{X} &\gate{S} &\gate{H} &\qw       \\
&\qw      &\qw              &\qw      &\qw        &\qw      &\qw      &\qw      &\qw      &\qw       \\ 
&\qw      &\qw              &\qw      &\qw        &\qw      &\qw      &\qw      &\qw      &\qw       \\
&\qw      &\qw              &\targ    &\gate{i Z} &\targ    &\qw      &\qw      &\qw      &\qw       \\
}}
\qquad
\subfloat{
    \Qcircuit @C=1em @R=.7em {
&\gate{H} &\gate{S^\dagger} &\gate{iZ}   &\gate{X} &\gate{S} &\gate{H} &\qw       \\
&\qw      &\qw              &\qw        &\qw      &\qw      &\qw      &\qw       \\ 
&\qw      &\qw              &\qw       &\qw      &\qw      &\qw      &\qw       \\
&\qw      &\qw              &\gate{ Z} &\qw      &\qw      &\qw      &\qw       \\
}}
\\
\subfloat{
    \Qcircuit @C=1em @R=.7em {
&\qw &\gate{Z}   &\qw       \\
&\qw &\qw       &\qw       \\ 
&\qw &\qw        &\qw       \\
&\qw &\gate{Z}   &\qw       \\
}}

    \caption{Diagonalization of $XXXY$}
    \label{fig:xxxy_diag}
\end{figure}

\end{document}